\documentclass[12pt, oneside]{amsart}

\usepackage[margin=1.25in]{geometry}
\usepackage{calrsfs}
\usepackage{setspace}
\usepackage{hyperref}
\usepackage{graphicx}
\usepackage[round]{natbib}
\usepackage{amsmath}
\usepackage{amsfonts}
\usepackage{amssymb}
\usepackage{amsthm}
\usepackage{enumerate}
\usepackage{xcolor}
\usepackage{tikz}

\definecolor{completeblue}{RGB}{30,66,110}
\definecolor{rpred}{RGB}{150,30,30}
\definecolor{transgreen}{RGB}{20,95,80}

\newtheorem{theorem}{Theorem}
\newtheorem{lemma}{Lemma}
\newtheorem{proposition}{Proposition}
\newtheorem{corollary}{Corollary}
\theoremstyle{definition}
\newtheorem{definition}{Definition}
\newtheorem{example}{Example}

\newcommand{\twoheaduparrow}{\text{\ooalign{$\uparrow$\cr\raisebox{0.35ex}{$\uparrow$}\cr}}}
\DeclareMathOperator*{\argmax}{arg\,max}
\newcommand{\uphull}[1]{{#1}^{\uparrow}}
\newcommand{\Uphull}[1]{{#1}^{\twoheaduparrow}}

\title[Existence of Richter--Peleg representation]{Existence of Richter--Peleg representation\\for general preferences}

\author{Leandro Gorno \hspace{30pt} Paulo K. Monteiro\vspace{2pt}\\
\hspace{16pt} \textit{FGV EPGE} \hspace{74pt} \textit{FGV EPGE} \hspace{8pt} \vspace{2pt}\\
}

\date{First draft: December 2021. Current version: \today.}

\begin{document}

\onehalfspacing

\begin{abstract}
We characterize the binary relations that admit a Richter--Peleg representation, imposing neither completeness nor transitivity. A relation admits such a representation if and only if it is strongly acyclic and its transitive closure is separable, where separability means embeddability in a preorder possessing a countable separating stratification---a countable family of pairwise disjoint subsets, none containing a strictly ranked pair, that resolves every strict comparison. Separating stratifications generalize Debreu's order-density condition, and for complete preorders our theorem reduces to Debreu's. The embedding clause is the exact price of the stratification structure: indispensable in general, as a Richter--Peleg representable partial order admitting no countable separating stratification shows, but unnecessary whenever the representation is continuous and the lower contour sets are connected, which covers many standard incomplete preferences in economics. By providing existence, our results operationalize the optimization theory of White (1980), which we also extend by delivering the set of maximal elements of each subset with a single Richter--Peleg representation and no completeness assumption. As a second application, we show that the impossibility theorem of Basu and Mitra (2003) on intergenerational equity is exactly a failure of separability.
\vspace{15pt}\vspace{10pt}
\end{abstract}

\maketitle

\vspace{5pt}

\section{Introduction}
\label{sec:introduction}

Utility plays a central role in economic theory. In the classical framework, a utility representation is a real-valued function that fully characterizes the underlying preference. This notion of utility imposes strong restrictions on preferences: any preference admitting such a representation must be both complete and transitive.

When decision-makers face complex or unfamiliar alternatives, however, it is often unreasonable to require their preferences to be complete or even transitive. These axioms also frequently fail when the decision-making agent is a collective entity such as a firm, committee, or society, because aggregation procedures can produce preferences violating completeness or transitivity even when those of all individual members satisfy these axioms. These observations have motivated an extensive research program on alternative representations accommodating incomplete preferences, intransitive indifference, and other departures from classical rationality.

One prominent alternative is the \textit{Richter--Peleg representation}, a real-valued function preserving both indifference and strict preference. Richter--Peleg representations have several appealing properties. They do not require the underlying binary relation to be complete or its indifference relation to be transitive. They encode the ordinal information needed to apply scalar optimization techniques: in particular, when a Richter--Peleg representation exists, any maximizer of the representing function is undominated according to the underlying relation \citep{White1980,Alcantud2016}. At the same time, as \citet{MajumdarSen1976} emphasized, a Richter--Peleg representation does not permit reconstruction of the represented relation because incomparability cannot be distinguished from indifference by inspecting the values of the function.\footnote{\citet{MajumdarSen1976} work with an asymmetric relation $P$. They call the condition $x \mathbin{P} y \Rightarrow f(x) > f(y)$ the \textit{Peleg--Richter representation} and name any $f$ satisfying it a \textit{transparent quasi-representation}. That condition is the strict-preference clause of Definition~\ref{def:RP} below, stated for a strict relation taken as primitive.}

Sufficient conditions for the existence of Richter--Peleg representations were established in classical work by \citet{Richter1966}, in the context of rational choice; by \citet{Peleg1970}, for (strict) partial orders on topological spaces; and by \citet{Aumann1962,Aumann1964}, for preferences over lotteries. Subsequent research extended these ideas to multi-utility representations, in which the preference is described by a set of utility functions rather than a single one \citep{Ok2002,Dubra2004,EvrenOk2011,Riella2015,Gorno2017}, and to representations consisting exclusively of Richter--Peleg utilities \citep{Alcantud2016}.

The question of necessary and sufficient conditions for the existence of a Richter--Peleg representation of a preorder has also been studied. \citet[Proposition~3.1]{Herden1989a} established that a preorder admits a Richter--Peleg representation if and only if it possesses a countable family of pairs of disjoint decreasing-increasing subsets covering all strict comparisons.\footnote{Herden's terminology differs from ours: what he calls a \textit{utility function} on a preordered set is an isotone $f$ with $f(x) < f(y)$ whenever $y$ is strictly preferred to $x$---that is, a Richter--Peleg representation in the sense of Definition~\ref{def:RP}. He calls his condition \textit{almost weak separability}.} Equivalent formulations in terms of countable separating families of increasing subsets appear in \citet{HackBraunGottwald2022,HackBraunGottwald2023}. These conditions are of a \emph{covering} type: the strict comparisons are caught collectively by a countable family of set-pairs, each strict pair witnessed by some member. Our condition is instead one of \emph{separation} in the sense of Debreu's order-denseness, which it directly generalizes: a single countable family resolves every strict comparison by interposing a stratum between its terms, and for singleton strata in the complete case this is exactly the existence of an intermediate point. The distinction is not merely formal---it is what makes our characterization specialize to Debreu's classical condition in the complete case (Proposition~\ref{proposition:order_density}).

A characterization for binary relations that may fail transitivity, however, has remained elusive. Representations of a different kind are available for such relations: \citet{NishimuraOk2016} obtain maxmin and minmax multi-utility representations for reflexive relations on metric spaces, continuous when the relation is continuous and the space is compact. Their representations are two-way, fully encoding the relation including its incomparabilities, and rely on topological structure. Ours is a one-way scalar Richter--Peleg representation, characterized by purely order-theoretic conditions. Moreover, the existing characterizations of Richter--Peleg representability, while structurally clean, do not have a direct conceptual parallel with the classic order-density condition of \citet{Debreu1954} for complete preorders---namely, the existence of a countable subset $Z \subseteq X$ such that every strict preference $x \succ y$ admits an intermediate $z \in Z$.

This paper closes both gaps. We provide a complete characterization of Richter--Peleg representability for arbitrary binary relations. Our main result (Theorem~\ref{theorem:main}) establishes that a binary relation $\succsim$ on a nonempty set $X$ admits a Richter--Peleg representation if and only if $\succsim$ is \textit{strongly acyclic} and its transitive closure is \textit{separable}.\footnote{The word ``separable'' is overloaded in this literature; see Definition~\ref{def:separability} for our precise usage, which differs from that of \citet{Herden1989a} (weak and almost weak separability), \citet{Peleg1970} (strict order-density), and \citet{Ok2002} (upper/lower separability), and from the topological notion. Our $\succsim$-density (Section~\ref{sec:debreu}) uses \emph{weak} interposition $x \succsim z \succsim y$, after \citet{Debreu1954}, whereas Peleg's and Ok's use strict interposition.} Strong acyclicity rules out weak-preference cycles containing a strict link; separability requires the relation to be embeddable in a preorder possessing a countable \textit{separating stratification}.

The notion of a stratification is the conceptual centerpiece of our framework. A stratification is a collection of pairwise disjoint subsets of $X$, each of which contains no pair of strictly ranked alternatives. Under completeness, the strata of a stratification reduce to subsets of indifference classes, so stratifications generalize indifference classes to settings where completeness fails. We further show (Proposition~\ref{proposition:order_density}) that, for complete preorders, the existence of a countable separating stratification coincides with the existence of a countable $\succsim$-dense subset of $X$ in Debreu's sense. Countable separating stratifications therefore generalize Debreu's classical density condition, with strata standing in for indifference classes.

Our framework is, for preorders, logically equivalent to that of \citet{Herden1989a} and \citet{HackBraunGottwald2022}: each yields the same class of representable preorders. \citet{Alcantud2016} establish a companion result of a different kind, showing that this class coincides with the class of preorders admitting Richter--Peleg \emph{multi}-utility representations. Three features of Theorem~\ref{theorem:main} distinguish our contribution.

The first is the extension beyond preorders to arbitrary binary relations. Strong acyclicity is the behavioral axiom that supports this extension: it ensures that passage to the transitive closure preserves all strict-preference information, so the analysis can be reduced to the preorder case. The scope of the extension deserves to be stated plainly. Strong acyclicity is weaker than transitivity, but it is not innocuous: it fails whenever a strict preference is linked by a chain of indifferences, so complete semiorders and the related models of intransitive indifference in the tradition of \citet{Luce1956} lie outside our reach (Section~\ref{sec:discussion}). What it accommodates is the case in which the recorded relation is a fragment of an underlying preorder---because the analyst observes only some comparisons, or an aggregation procedure returns only some---and it is exactly the condition under which a scalar index consistent with every recorded comparison can exist.

The second is the structural parallel with Debreu's classical theory. The existing conditions ask for a countable family of increasing sets, unconstrained relative to one another; ours asks for a family that is in addition pairwise disjoint and rank-coherent. That extra structure is what a stratum needs in order to behave like an indifference class, and it is what makes the condition specialize, under completeness, to Debreu's order-density, and Theorem~\ref{theorem:main} to Debreu's representation theorem (Proposition~\ref{proposition:order_density}). Theorem~\ref{theorem:impossibility} measures its cost.

The third is an impossibility result (Theorem~\ref{theorem:impossibility}) exhibiting a partial order that admits a Richter--Peleg representation but no countable separating stratification. This shows that the embedding clause in our definition of separability is not a matter of mathematical convenience: it reflects a structural feature of Richter--Peleg representability.

Two corollaries highlight the usefulness of Theorem~\ref{theorem:main}. When $\succsim$ is a preorder (Corollary~\ref{corollary:preorder}), it admits a Richter--Peleg representation if and only if it is separable. When $X$ is countable (Corollary~\ref{corollary:countable}), strong acyclicity alone suffices.

Section~\ref{sec:application} develops two applications. The first refines the optimization theory of Richter--Peleg representations. Building on \citet{White1980}, we show (Proposition~\ref{proposition:optimality1}) that, when a preorder $\succsim$ admits a Richter--Peleg representation, for each subset $A \subseteq X$ there is a Richter--Peleg representation of $\succsim$ whose maximizers over $A$ are exactly the maximal elements of $\succsim$ in $A$; and (Proposition~\ref{proposition:optimality2}) that every alternative maximal in some set is the maximizer, over every such set simultaneously, of a single Richter--Peleg representation. These results extend White's theorem, which delivers a single representation only under completeness, and Example~\ref{example:A_dependence} shows that the resulting dependence on $A$ in Proposition~\ref{proposition:optimality1} cannot be dispensed with. The second application concerns intergenerational equity: a social welfare function satisfying strong Pareto and finite anonymity is precisely a Richter--Peleg representation of the Suppes--Sen grading preorder, so the impossibility theorem of \citet{BasuMitra2003} says exactly that this preorder is not separable.

The remainder of the paper is organized as follows. Section~\ref{sec:preliminaries} collects preliminaries. Section~\ref{sec:RPRepresentation} defines the Richter--Peleg representation, presents basic properties and strong acyclicity, and relates strong acyclicity to other cycle-restricting notions in the literature. Section~\ref{sec:separability} introduces stratifications, embeddings, and the formal definition of separability. Section~\ref{sec:main_result} contains the main theorem, the comparison between countable separating stratifications and Debreu's density condition, the impossibility result establishing the necessity of the embedding clause, the relation to existing characterizations for preorders, and the corollaries. Section~\ref{sec:application} develops the applications. Section~\ref{sec:discussion} concludes with a discussion of the reach of our results, related revealed-preference considerations, and continuity. All proofs are in the Appendix.

\section{Preliminaries}
\label{sec:preliminaries}

Let $X$ be a nonempty set. A \textit{binary relation on} $X$ is any subset $\succsim \, \subseteq \, X \times X$. As usual, for any $x,y \in X$, we write $(x,y) \in \hspace{3pt} \succsim$ as $x \succsim y$. The \textit{symmetric}, \textit{asymmetric}, and \textit{incomparable} parts of $\succsim$ are defined as
\[
\sim \hspace{2pt} := \left\{(x,y) \in X \times X \, \middle| \, x \succsim y \text{ and } y \succsim x\right\},
\]
\[
\succ \hspace{2pt} := \left\{(x,y) \in X \times X \, \middle| \, x \succsim y \text{ and not } y \succsim x\right\},
\]
\[
\bowtie \hspace{2pt} := \left\{(x,y) \in X \times X \, \middle| \, \text{neither } x \succsim y \text{ nor } y \succsim x\right\},
\]
respectively.

The binary relation $\succsim$ is \textit{reflexive} if $x \succsim x$ for all $x \in X$; \textit{complete} if $x \succsim y$ or $y \succsim x$ for all $x, y \in X$; \textit{transitive} if $x \succsim y$ and $y \succsim z$ imply $x \succsim z$ for all $x,y,z \in X$; and \textit{antisymmetric} if $x \succsim y$ and $y \succsim x$ imply $x = y$. A \textit{preorder} is a reflexive and transitive binary relation, and a \textit{partial order} is an antisymmetric preorder.\footnote{Some authors, including \citet{Peleg1970}, reserve ``partial order'' for an irreflexive and transitive relation, that is, for what we would call the asymmetric part of a partial order. We follow the reflexive convention throughout.} A set $B \subseteq X$ is \textit{increasing} if $x \in B$ and $y \succsim x$ imply $y \in B$, and \textit{decreasing} if its complement is increasing. A function $u : X \to \mathbb{R}$ is \textit{isotone} if $x \succsim y$ implies $u(x) \geq u(y)$.

We use repeatedly the following elementary consequence of transitivity, applied below both to $\succsim$ and to other transitive relations: if $\succsim$ is transitive and $x,y,z \in X$, then 
\begin{equation} \label{eq:strict_composition} x \succsim y \succ z \quad \text{or} \quad x \succ y \succsim z \qquad \Longrightarrow \qquad x \succ z. 
\end{equation} 
Transitivity gives $x \succsim z$ in either case, while $z \succsim x$ would give $z \succsim y$ in the first and $y \succsim x$ in the second.

The \textit{transitive closure} of $\succsim$, denoted $\succsim^t$, is the smallest preorder containing $\succsim$; explicitly, $x \succsim^t y$ if and only if $x = y$, $x \succsim y$, or there exist $x_1, \ldots, x_n \in X$ with $x \succsim x_1 \succsim x_2 \succsim \ldots \succsim x_n \succsim y$. By construction, $\succsim^t$ is reflexive and transitive and satisfies $\succsim \, \subseteq \, \succsim^t$. Imposing reflexivity in this way is innocuous for our purposes: by Lemma~\ref{lemma:reflexivity} below, adjoining the diagonal to a relation leaves its set of Richter--Peleg representations unchanged. The relations $\sim^t$ and $\succ^t$ denote the symmetric and asymmetric parts of $\succsim^t$.

\section{Richter--Peleg representation}
\label{sec:RPRepresentation}

\subsection{Definition and basic properties}
\label{sec:RP_definition}

\begin{definition}[Richter--Peleg representation]
\label{def:RP}
Let $\succsim$ be a binary relation on a nonempty set $X$. A function $u: X\to \mathbb{R}$ is a \textit{Richter--Peleg representation} of $\succsim$ if it satisfies:
\begin{enumerate}
\item $x\sim y$ implies $u(x)=u(y)$, and
\item $x \succ y$ implies $u(x)>u(y)$,
\end{enumerate}
for all $x,y \in X$.
\end{definition}

A Richter--Peleg representation $u$ of $\succsim$ preserves both indifference and strict preference, and is in particular isotone: $x \succsim y$ means $x \sim y$ or $x \succ y$, so $u(x) \geq u(y)$ in either case. The representation does not, however, in general allow $\succsim$ to be recovered from $u$ alone. In particular, the function values cannot distinguish indifferent pairs from incomparable pairs: if $u(x) = u(y)$, the underlying relation may have either $x \sim y$ or $x \bowtie y$, and the function carries no information sufficient to discriminate between these two cases. In this sense, \citet{MajumdarSen1976} called Richter--Peleg representations ``transparent quasi-representations'' since, while they respect the strict preference structure, they do not fully encode the binary relation.

This information loss has motivated a complementary line of research on \textit{multi-utility representations}, in which a preorder is represented by a family $\mathcal{U}$ of real-valued functions such that $x \succsim y$ if and only if $u(x) \geq u(y)$ for every $u \in \mathcal{U}$ \citep{Ok2002,EvrenOk2011}. A multi-utility representation fully encodes $\succsim$ and allows incomparable pairs to be detected through disagreement among the representing functions. The two approaches are complementary: a multi-utility representation encodes more information but lacks the immediate scalar-optimization interpretation of a Richter--Peleg representation. \citet[Theorem~3.1]{Alcantud2016} show that, for preorders, the existence of a Richter--Peleg representation is equivalent to the existence of a multi-utility representation each of whose members is itself a Richter--Peleg representation---a \emph{Richter--Peleg multi-utility}; the optimization results of Section~\ref{sec:application} exploit the Richter--Peleg property to reduce maximization of the underlying preorder to scalar optimization.

\subsection{Strong acyclicity}
\label{sec:strong_acyclicity}

The existence of a Richter--Peleg representation does not require $\succsim$ to be transitive, as the following example shows.

\begin{example}
\label{example:transitivity}
Let $X=\{a,b,c\}$ and define $\succsim \hspace{2pt} := \left\{(a,a),(b,b),(c,c),(a,b),(b,c)\right\}$. Then $\succsim$ is not transitive. However, the function $u:X \to \mathbb{R}$ given by $u(a)=2$, $u(b)=1$, and $u(c)=0$ is a Richter--Peleg representation of $\succsim$.
\end{example}

Richter--Peleg representability does, however, impose a definite restriction on the structure of $\succsim$. The function $u$ in Example~\ref{example:transitivity} respects the strict comparisons because $\succsim$ contains no cycles connecting strictly ranked alternatives. The following condition captures this requirement.

\begin{definition}[Strong acyclicity]
\label{def:strong_acyclicity}
A binary relation $\succsim$ on a set $X$ is \textit{strongly acyclic} if $x_1 \succsim x_2 \succsim \ldots \succsim x_n \succsim x_1$ implies $x_1 \sim x_2 \sim \ldots \sim x_n \sim x_1$ for all $x_1, x_2, \ldots, x_n \in X$.
\end{definition}

\begin{lemma}
\label{lemma:strong_acyclicity}
If a binary relation $\succsim$ admits a Richter--Peleg representation, then $\succsim$ is strongly acyclic.
\end{lemma}

Violating strong acyclicity thus provides a simple way to construct binary relations that fail to admit a Richter--Peleg representation:

\begin{example}
\label{example:strict_cycle}
Let $X = \{a,b,c\}$ and
\[
\succsim \hspace{2pt} = \left\{(a,a),(b,b),(c,c),(a,b),(b,c),(c,a)\right\}.
\]
$\succsim$ is not strongly acyclic and thus cannot admit a Richter--Peleg representation.
\end{example}

Strong acyclicity is necessary but not sufficient for the existence of a Richter--Peleg representation. The following classic example shows that even a complete and transitive binary relation on a compact and connected space may fail to admit such a representation.

\begin{example}[Lexicographic preferences]
\label{example:lex}
Let $X = [0,1]\times [0,1]$ and
\[
\succsim \hspace{2pt} = \left\{((x_1, x_2), (y_1, y_2)) \in X \times X \, \middle| \, x_1 > y_1 \vee \left[x_1 = y_1 \wedge x_2 \geq y_2\right]\right\}.
\]
Then $\succsim$ cannot admit a Richter--Peleg representation. Were $u$ one, then for each $t \in [0,1]$ the open interval $I_t := \big(u(t,0), u(t,1)\big)$ would be nonempty because $(t,1) \succ (t,0)$, and the family $\{I_t\}_{t \in [0,1]}$ would be pairwise disjoint because $t < t'$ implies $(t',0) \succ (t,1)$. But an uncountable family of pairwise disjoint nonempty open intervals in $\mathbb{R}$ is impossible, since each contains a distinct rational.
\end{example}

\subsection{Other notions of acyclicity}
\label{sec:related_acyclicity}

Strong acyclicity is one of several cycle-restricting notions used in the literature on binary relations. We review the most relevant alternatives, collect them into an equivalence result, and discuss our reasons for working with strong acyclicity.

A binary relation $\succsim$ is \textit{acyclic} if $x_1 \succ x_2 \succ \ldots \succ x_n$ implies $x_n \not\succ x_1$ for all $x_1, \ldots, x_n \in X$ \citep{Walker1977}. Acyclicity rules out cycles of strict preference but admits cycles in $\succsim$ that contain a strict link as long as that link lies on a cycle threaded through indifferent alternatives. Acyclicity is therefore strictly weaker than strong acyclicity:

\begin{example}
\label{example:acyclic_not_strong}
Let $X = \{a, b, c\}$ and \[ \succsim \hspace{2pt} := \left\{(a,a),(b,b),(c,c),(a,b), (b,a), (b,c), (c,b), (c,a)\right\}. \] Then $\succ = \{(c,a)\}$, so $\succsim$ is acyclic. However, $a \succsim b \succsim c \succsim a$ while $c \not\sim a$, so $\succsim$ is not strongly acyclic.
\end{example}

The difference between acyclicity and strong acyclicity vanishes for antisymmetric relations, where the symmetric part $\sim$ is trivial. The two conditions diverge precisely when nontrivial indifference classes are present.

\citet{Suzumura1976} introduced the notion of \textit{consistency} of a binary relation. A binary relation $\succsim$ is \textit{Suzumura consistent} if there is no chain $x_1 \succ x_2 \succsim x_3 \succsim \ldots \succsim x_n \succsim x_1$. Suzumura consistency rules out cycles in $\succsim$ that include at least one strict link. \citet{Suzumura1976} shows (his Theorem~3) that consistency is equivalent to the existence of a complete preorder extending $\succsim$ and preserving its strict preference structure; the book-length treatment is \citet{BossertSuzumura2010}.

\citet[Chapter~2]{ChambersEchenique2016} consider an \emph{order pair} $(R, P)$ with $P \subseteq R$ and call it \emph{acyclic} if there is no chain $x_1 \mathbin{R} x_2 \mathbin{R} \ldots \mathbin{R} x_n \mathbin{P} x_1$. Specializing to $R := \, \succsim$ and $P := \, \succ$ gives a condition on a single binary relation. To avoid collision with the unqualified ``acyclic'' used above, we refer to this property of $(\succsim,\succ)$ as \emph{pair-acyclicity}. 

A fourth condition is structural rather than behavioral. The transitive closure $\succsim^t$ of $\succsim$ always satisfies $\succsim \, \subseteq \, \succsim^t$. One can ask in addition whether $\succsim^t$ is an \textit{extension} of $\succsim$, in the sense that $x \succ y$ implies $x \succ^t y$ for all $x, y \in X$.

The four conditions just introduced---strong acyclicity, Suzumura consistency, pair-acyclicity of $(\succsim, \succ)$, and the extension property of the transitive closure---are mutually equivalent.

\begin{proposition}
\label{proposition:equivalence}
For any binary relation $\succsim$ on a nonempty set $X$, the following statements are equivalent:
\begin{enumerate}
\item $\succsim$ is strongly acyclic;
\item $\succsim$ is Suzumura consistent;
\item $(\succsim, \succ)$ is pair-acyclic;
\item the transitive closure $\succsim^t$ is an extension of $\succsim$.
\end{enumerate}
\end{proposition}

Suzumura consistency and pair-acyclicity differ only in where the strict link is placed along the cycle, so their equivalence is immediate; the substantive content of Proposition~\ref{proposition:equivalence} is that both coincide with strong acyclicity and with the extension property of the transitive closure. The four formulations nonetheless offer distinct perspectives on the same feature: the first two make the asymmetric part $\succ$ explicit, the transitive-closure formulation is purely structural, and strong acyclicity is stated directly as a property of the single primitive relation $\succsim$. We work with strong acyclicity in what follows because the primitive of choice theory is typically a single binary relation, and a condition stated directly on that relation is the most useful for axiomatic analysis. The equivalence above ensures that this choice entails no loss of generality.

\section{Separable binary relations}
\label{sec:separability}

This section introduces the structural concepts used in the statement of our main result.

\subsection{Stratifications}
\label{sec:stratifications}

\begin{definition}[Stratification]
\label{def:stratification}
Let $\succsim$ be a binary relation on a nonempty set $X$. A collection $\mathcal{A}$ of nonempty, pairwise disjoint subsets of $X$ is a \textit{stratification} of $\succsim$ if, for every $A, A' \in \mathcal{A}$, $x, y \in A$, and $x', y' \in A'$, $x \succ x'$ implies $y' \not\succ y$.
\end{definition}

The definition implies that $\succsim$ induces an asymmetric relation among the strata: if some element of $A$ is strictly preferred to some element of $A'$, then no element of $A'$ is strictly preferred to any element of $A$. In particular, no stratum contains a pair of strictly ranked alternatives.\footnote{This follows from Definition~\ref{def:stratification} by setting $A=A'$, $x=y'$, and $y=x'$, so that $x \succ y$ implies $x \not\succ y$.} The strata thus play the role of indifference classes in a setting where indifference need not be transitive: each stratum gathers alternatives that the relation treats as ``of the same rank,'' even when they are not literally indifferent.

A stratification can be very coarse---taking the whole of $X$ as a single stratum is admissible whenever the resulting set contains no strict pair---or as fine as the family of all singletons. To be useful, however, a stratification must witness the strict preferences of $\succsim$.

\begin{definition}[Separating stratification]
\label{def:separating_stratification}
A stratification $\mathcal{A}$ of $\succsim$ is \textit{separating} if, for all $x, y \in X$ with $x \succ y$, there are $A \in \mathcal{A}$ and $z \in A$---a \textit{witness} for the pair $(x,y)$---such that $x \succsim z$, and moreover $y \not\succ z'$ for every $z' \in A$.
\end{definition}

Separation is most transparent in terms of increasing hulls. For $A \subseteq X$, let
\[
\uphull{A} := \left\{x \in X \, \middle| \, x \succsim z \text{ for some } z \in A\right\}, \qquad \Uphull{A} := \left\{x \in X \, \middle| \, x \succ z \text{ for some } z \in A\right\}
\]
denote the \textit{weak} and \textit{strict increasing hulls} of $A$. Then $\mathcal{A}$ is separating if and only if, for every $x \succ y$, some $A \in \mathcal{A}$ satisfies $x \in \uphull{A}$ and $y \notin \Uphull{A}$. This is the natural generalization of Debreu's order-denseness: for a singleton stratum $A = \{z\}$ in a complete preorder the condition reduces to $x \succsim z$ and $z \succsim y$---the interposition of an intermediate point between $x$ and $y$---and a separating stratification lifts this requirement to relations that need not be complete, with strata in place of single interposed points. It is essential that $y \not\succ z'$ hold for \emph{every} $z' \in A$, not merely for a witness: were it required of the witness alone, a stratum could separate through a good witness $z$ while harboring some $w$ with $y \succ w$, which places $y$ in $\Uphull{A}$ and defeats the representation built in the proof of Theorem~\ref{theorem:main}. Finally, when $\succsim$ is transitive both $\uphull{A}$ and $\Uphull{A}$ are increasing sets, so a countable separating stratification $\{A_i\}$ induces the countable family of increasing sets $\left\{\uphull{A_i}\right\} \cup \left\{\Uphull{A_i}\right\}$ discussed in Section~\ref{sec:related_characterizations}.

Two examples fix intuition: the first shows that the separating requirement has real content, the second that separating stratifications always exist in a trivial form.

\begin{example}
\label{example:two_stratifications}
Let $X=\{a,b,c\}$ and let $\succsim \hspace{2pt} := \left\{(a,c),(b,c),(a,a),(b,b),(c,c)\right\}$. Then both $\mathcal{A}_1 := \left\{\{a\}\right\}$ and $\mathcal{A}_2 := \left\{\{a,b\}\right\}$ are stratifications of $\succsim$. However, only $\mathcal{A}_2$ is separating: it separates $a \succ c$ (taking $z = a$, since $c \not\succ a$ and $c \not\succ b$) and $b \succ c$ (taking $z = b$), whereas $\mathcal{A}_1$ fails to separate $b \succ c$, which would require $b \succsim a$.
\end{example}

\begin{example}
\label{example:separating_singletons}
For any binary relation $\succsim$, the collection $\left\{\{x\} \, \middle| \, x \in X\right\}$ of all singletons is a separating stratification of $\succsim$. It is a stratification: the singletons are pairwise disjoint, and for $A = \{a\}$ and $A' = \{a'\}$ the requirement of Definition~\ref{def:stratification} reads ``$a \succ a'$ implies $a' \not\succ a$'', which is asymmetry of $\succ$. It is separating: given $x \succ y$, the singleton $\{y\}$ satisfies $x \succsim y$ and $y \not\succ y$, the latter by asymmetry of $\succ$.
\end{example}

The role of separating stratifications is to summarize the order structure of $\succsim$ in a more manageable object. Thus, the singleton stratification of the last example is rarely useful on its own because it is as large as $X$.

\subsection{Embeddings and separability}
\label{sec:embeddings}

Example~\ref{example:separating_singletons} shows that every reflexive binary relation admits a separating stratification. The characterization developed in the next section will, however, require the existence of a separating stratification that is also countable. As we shall see, ensuring that a countable separating stratification exists may require ``enriching'' $\succsim$ with additional alternatives. Embeddings provide the formal device for this enrichment.

\begin{definition}[Embedding]
\label{def:embedding}
Let $X$ and $X^*$ be nonempty sets and let $\succsim$ and $\succsim^*$ be binary relations on $X$ and $X^*$, respectively. We say that $\succsim^*$ is an \textit{embedding} of $\succsim$ if there exists a function $f: X \to X^*$ such that $x \succsim y$ if and only if $f(x) \succsim^* f(y)$ for all $x, y \in X$.
\end{definition}

\begin{definition}[Separability]
\label{def:separability}
A binary relation $\succsim$ is \textit{separable} if it can be embedded in a preorder with a countable separating stratification.
\end{definition}

The embedding enlarges the domain with auxiliary alternatives, facilitating the separation of all strict comparisons with countably many strata.

\begin{example}
\label{example:separable}
Let $X = \mathbb{R} \setminus \mathbb{Q}$ be the set of irrational numbers and let $\succsim$ be the natural order on $X$. Taking $X^* = \mathbb{R}$ and $\succsim^*$ the natural order on $X^*$, the relation $\succsim^*$ is an embedding of $\succsim$ via the identity function, and $\left\{\{q\} \, \middle| \, q \in \mathbb{Q}\right\}$ is a countable separating stratification of $\succsim^*$. Hence $\succsim$ is separable.
\end{example}

Separability is a demanding requirement: it forces the relation to be a preorder.

\begin{lemma}
\label{lemma:transitivity_embedding}
Every binary relation that can be embedded in a preorder is a preorder.
\end{lemma}

In particular, every separable binary relation is a preorder, and hence strongly acyclic (for any weak-preference cycle through a preorder collapses to indifference).

\section{Main result}
\label{sec:main_result}

The main result of this paper provides general necessary and sufficient conditions for the existence of a Richter--Peleg representation.

\begin{theorem}
\label{theorem:main}
Let $X$ be a nonempty set and let $\succsim$ be a binary relation on $X$. Then $\succsim$ admits a Richter--Peleg representation if and only if $\succsim$ is strongly acyclic and the transitive closure of $\succsim$ is separable.
\end{theorem}

Theorem~\ref{theorem:main} characterizes Richter--Peleg representability by two conditions of distinct character. The behavioral condition, strong acyclicity, rules out the cycles incompatible with any scalar representation of strict preference (Lemma~\ref{lemma:strong_acyclicity}). Given strong acyclicity, the analysis reduces to the preorder structure of the transitive closure, and the structural condition, separability, demands that this preorder be embeddable in another preorder admitting a countable family of strata that separate all strict pairs. The embedding clause is essential: as Theorem~\ref{theorem:impossibility} below shows, separability of a preorder is strictly weaker than the requirement that the preorder admits a countable separating stratification.

\subsection{Stratifications and order-denseness}
\label{sec:debreu}

The existence of a countable separating stratification specializes, under completeness and transitivity, to a familiar condition in classical utility theory due to \citet{Debreu1954}. Recall that a subset $Z \subseteq X$ is \textit{$\succsim$-dense} if, for all $x, y \in X$ with $x \succ y$, there exists $z \in Z$ such that $x \succsim z \succsim y$.\footnote{The weak interposition $x \succsim z \succsim y$ is Debreu's. \citet{Peleg1970} and \citet{Ok2002} use the strict variant $x \succ z \succ y$, and \citet{Ok2002} attaches the label ``$\succsim$-dense'' to it; the two differ exactly on jumps.}

\begin{proposition}
\label{proposition:order_density}
Let $\succsim$ be a binary relation on a nonempty set $X$.
\begin{enumerate}[(i)]
\item If $Z \subseteq X$ is a $\succsim$-dense subset, then the collection of singletons $\left\{\{z\} \, \middle| \, z \in Z\right\}$ is a separating stratification of $\succsim$.
\item If, in addition, $\succsim$ is a complete preorder, $\mathcal{A}$ is a separating stratification of $\succsim$, and $z_A \in A$ is an arbitrary representative for each $A \in \mathcal{A}$, then $\left\{z_A \, \middle| \, A \in \mathcal{A}\right\}$ is a $\succsim$-dense subset of $X$.
\end{enumerate}
Consequently, for a complete preorder the least cardinality of a separating stratification of $\succsim$ equals the least cardinality of a $\succsim$-dense subset of $X$.
\end{proposition}

Proposition~\ref{proposition:order_density} formalizes a back-and-forth construction between $\succsim$-dense subsets and separating stratifications that is free of any cardinality loss in either direction: the singletons of a $\succsim$-dense set are a separating stratification, and one representative per stratum recovers a $\succsim$-dense set. In particular, a complete preorder admits a countable separating stratification if and only if $X$ contains a countable $\succsim$-dense subset. Combined with Theorem~\ref{theorem:main} and Debreu's representation theorem \citep{Debreu1954,Fishburn1970}, this means that our characterization reduces exactly to Debreu's for complete preorders, since under completeness a Richter--Peleg representation is a utility representation in the usual sense. The reduction is structural rather than merely extensional: as noted in Section~\ref{sec:stratifications}, for a singleton stratum in a complete preorder the separating condition is exactly Debreu's interposition of an intermediate point.

Part~(i) uses neither completeness nor transitivity, so a countable $\succsim$-dense subset is sufficient for a countable separating stratification---and hence, for preorders, for Richter--Peleg representability---quite generally; only the converse needs completeness. Under incompleteness the two conditions genuinely part company: Theorem~\ref{theorem:impossibility} below exhibits a representable partial order with no countable $\succsim$-dense subset. 

The embedding clause in our definition of separability, essential in general, is redundant in the complete-preorder case: when $\succsim$ is a complete preorder, both representability and Debreu's order-denseness refer only to the relation on $X$ itself.

\subsection{The role of the embedding clause}
\label{sec:impossibility}

The definition of separability invokes both a stratification and an embedding. The next result shows that the embedding clause cannot be omitted, even for partial orders.

\begin{theorem}
\label{theorem:impossibility}
There exists a partial order $\succsim$ on a set $X$ that admits a Richter--Peleg representation but no countable separating stratification.
\end{theorem}

The construction takes $X$ to be the set of countable ordinals, equipped with the partial order defined by $x \succsim y$ if and only if $\varphi(x) \geq \varphi(y)$ and $x \preceq^* y$, where $\preceq^*$ is the natural well-order on the ordinals and $\varphi: X \to [0,1]$ is any injection. Thus $\succsim$ is the intersection of the $\varphi$-order with the \emph{reverse} of $\preceq^*$, and $\varphi$ is a Richter--Peleg representation of it. That no countable separating stratification exists rests on a cardinality argument; full details are in Appendix~\ref{appendix:proofs}.

The obstruction is not a lack of comparabilities but a mismatch of cardinality. Comparability requires the two orders to disagree, so incomparability requires them to agree: if $a$ and $b$ are incomparable with $\varphi(a) > \varphi(b)$, then $b \prec^* a$. On a stratum, which contains no strictly ranked pair, $\varphi$ is therefore isotone for the well-order, and an injective isotone map from a well-ordered set into $\mathbb{R}$ cannot have uncountable domain, again by the argument of Example~\ref{example:lex}. Every stratum is thus countable, the union of countably many of them is bounded in the well-order, and the strict comparisons lying above that bound are left with no witness below them. Embedding the relation in a richer preorder supplies the additional alternatives needed to form the missing strata.

This is precisely the role played by the \textit{shadow completion} used in our proof of Theorem~\ref{theorem:main}, illustrated in Figure~\ref{fig:shadow}. The obstacle to separating $x \succ y$ inside $X$ is the absence of a suitable interposed alternative: one ranked below $x$ but not above $y$. The shadow completion manufactures them. At every rational level $q$ strictly below $u(x)$ it adjoins a copy $(x,q)$ of $x$, ranked below $x$ and carrying utility $q$; copies of distinct originals sharing a level are left mutually incomparable, so the copies at level $q$ assemble into a single stratum $A_q$ containing no strictly ranked pair. Choosing $q$ between $u(y)$ and $u(x)$ then separates the pair, since $x$ lies above $(x,q) \in A_q$ while $y$ lies above nothing in $A_q$. Countably many levels suffice because $\mathbb{Q}$ is dense in $\mathbb{R}$, however large $X$ may be; this is precisely what no stratification of the original relation could achieve. This level construction is related in form---though not in purpose---to the formal-ball model used to represent metric spaces as continuous posets \citep{EdalatHeckmann1998,GierzEtAl2003}; we return to the comparison in the proof of Lemma~\ref{lemma:preorder_case}.

\begin{figure}[htbp]
\centering
\begin{tikzpicture}[scale=1.1]
\draw[->,black!50,thick] (0,0) -- (0,6.3);
\node[black!50] at (0,6.6) {$u^*$};
\draw[black!50,thick] (-0.15,5.4) -- (0.15,5.4);
\node[left,black!50] at (-0.15,5.4) {$u(x)$};
\draw[rpred,thick] (-0.15,3.2) -- (0.15,3.2);
\node[left,rpred] at (-0.15,3.2) {$r$};
\draw[black!50,thick] (-0.15,1.6) -- (0.15,1.6);
\node[left,black!50] at (-0.15,1.6) {$u(y)$};
\draw[dashed,rpred,thick] (0.9,3.2) -- (8.9,3.2);
\node[rpred,right] at (8.9,3.2) {$A_r$};
\draw[dotted,black!45] (2.3,0.3) -- (2.3,5.05);
\foreach \h in {0.5,1.05,1.95,2.6,3.7,4.85}
\node[circle,draw=black!55,fill=white,inner sep=0pt,minimum size=3.6pt] at (2.3,\h) {};
\node[rectangle,fill=rpred,draw=rpred,inner sep=0pt,minimum size=5pt] at (2.3,3.2) {};
\node[circle,fill=completeblue,draw=completeblue,inner sep=0pt,minimum size=6pt] at (2.3,5.4) {};
\node[above,completeblue] at (2.3,5.65) {$x$};
\draw[->,rpred,thin] (2.65,5.05) -- (2.65,3.35);
\node[rpred] at (3.35,3.55) {\scriptsize $(x,r)$};
\draw[dotted,black!45] (4.8,0.3) -- (4.8,3.95);
\foreach \h in {0.5,1.05,1.95,2.6,3.7}
\node[circle,draw=black!55,fill=white,inner sep=0pt,minimum size=3.6pt] at (4.8,\h) {};
\node[rectangle,fill=rpred,draw=rpred,inner sep=0pt,minimum size=5pt] at (4.8,3.2) {};
\node[circle,fill=completeblue,draw=completeblue,inner sep=0pt,minimum size=6pt] at (4.8,4.3) {};
\node[above,completeblue] at (4.8,4.55) {$z$};
\node[rpred] at (5.6,3.55) {\scriptsize $(z,r)$};
\node[rpred,fill=white,inner sep=1.5pt] at (3.55,3.2) {\scriptsize $\bowtie$};
\draw[dotted,black!45] (7.6,0.3) -- (7.6,1.25);
\foreach \h in {0.5,1.05}
\node[circle,draw=black!55,fill=white,inner sep=0pt,minimum size=3.6pt] at (7.6,\h) {};
\node[circle,fill=completeblue,draw=completeblue,inner sep=0pt,minimum size=6pt] at (7.6,1.6) {};
\node[above,completeblue] at (7.6,1.85) {$y$};
\node[black!55,fill=white,inner sep=1.5pt] at (7.6,3.2) {\scriptsize $\times$};
\node[black!55] at (7.6,4.0) {\tiny no shadow of $y$};
\end{tikzpicture}
\caption{The shadow completion, for $x \succ y$ and a rational level $r$ with $u(y) < r < u(x)$. Filled circles are originals; open circles are shadows, one $(x,q)$ at level $u^*(x,q)=q$ for every rational $q < u(x)$. Squares mark the shadows at level $r$, which together form the stratum $A_r$; distinct originals' shadows at a common level are $\succsim^*$-incomparable ($\bowtie$), so $A_r$ contains no strictly ranked pair. Moreover, $A_r$ separates $x \succ y$ because $x \succsim^* (x,r) \in A_r$ and, for any $(w,r) \in A_r$, we must have $u(w) > r > u(y)$, thus $y \not\succsim^* (w,r)$.}
\label{fig:shadow}
\end{figure}

Theorem~\ref{theorem:impossibility} shows that the embedding cannot be dispensed with in general. It can be dispensed with often, and the reason is visible in the construction: what the shadow completion supplies is interposed alternatives at prescribed utility levels, and many domains of economic interest contain them already. The starting point is that a Richter--Peleg representation always partitions its domain into strata---its level sets, which satisfy Definition~\ref{def:stratification} for the trivial reason that $u$ is strictly increasing along $\succ$, and which are exactly the indifference classes when $\succsim$ is complete. There are in general uncountably many of them. The next result gives a condition under which countably many already suffice.

\begin{proposition}
\label{proposition:no_embedding}
Let $X$ be a topological space and let $\succsim$ be a reflexive binary relation on $X$ admitting a continuous Richter--Peleg representation $u$. If the lower contour set $L(x) := \left\{z \in X \, \middle| \, x \succsim z\right\}$ is connected for every $x \in X$, then
\[
\left\{u^{-1}(q) \, \middle| \, q \in \mathbb{Q}\right\} \setminus \left\{\emptyset\right\}
\]
is a countable separating stratification of $\succsim$.
\end{proposition}

When $\succsim$ is a preorder, Proposition~\ref{proposition:no_embedding} implies that it is separable with no enlargement of its domain. The hypotheses are sufficient and far from necessary---a countable domain requires none of them---but they are met by the many incomplete preferences used in economics.\footnote{It is essential that the hypothesis concern the lower contour sets rather than $X$ itself. Connectedness of $X$ does not suffice: under the continuum hypothesis, the relation constructed in the proof of Theorem~\ref{theorem:impossibility} can be realized on $X = [0,1]$ with its usual topology and with $\varphi$ the identity, which is continuous, so that $X$ is a compact interval and yet every separating stratification of $\succsim$ is uncountable. A path from $y$ to $x$ may leave $L(x)$, and the intermediate values of $u$ are then attained at alternatives that $x$ does not dominate.}

\begin{example}
\label{example:vector_order}
The natural vector order on a convex $X \subseteq \mathbb{R}^n$ has convex, hence connected, lower contour sets, and $u(x) = \sum_i x_i$ is a continuous Richter--Peleg representation of it. Proposition~\ref{proposition:no_embedding} therefore applies, and no embedding is required.
\end{example}

\begin{example}
\label{example:mixture}
A \textit{DMO preference} is a preorder on the set $P(Z)$ of Borel probability measures over a compact metric prize space $Z$ that satisfies independence\footnote{A binary relation $\succsim$ on $P(Z)$ satisfies \textit{independence} if, for all $p,q,r \in P(Z)$ and $\alpha \in [0,1]$, $p \succsim q$ is equivalent to $\alpha p +(1-\alpha) r \succsim \alpha q +(1-\alpha) r$.} and is closed as a subset of $P(Z) \times P(Z)$, where $P(Z)$ carries the topology of weak convergence \citep{Dubra2004}. Independence makes every lower contour set convex, hence connected: if $p \succsim q_1$ and $p \succsim q_2$, then $p \succsim \lambda p + (1-\lambda) q_2 \succsim \lambda q_1 + (1-\lambda) q_2$ for every $\lambda \in [0,1]$. And \citet[Proposition~3]{Dubra2004} show that every DMO preference admits a continuous \textit{Aumann utility}: a function $v \in C(Z)$ whose expectation $p \mapsto \int_Z v \, dp$ is a Richter--Peleg representation of $\succsim$, and is automatically continuous for the topology of weak convergence. Proposition~\ref{proposition:no_embedding} therefore applies here too, and no embedding is required. First- and second-order stochastic dominance on the lotteries over a compact interval are instances.
\end{example}

\subsection{Stratification numbers}
\label{sec:stratification_numbers}

Theorem~\ref{theorem:main}, Theorem~\ref{theorem:impossibility} and Proposition~\ref{proposition:order_density} are statements about a single cardinal, and naming it clarifies how they fit together.

\begin{definition}
\label{def:stratification_number}
The \textit{stratification number} of a binary relation $\succsim$ on $X$ is
\[
\mathfrak{s}(\succsim) := \min\left\{|\mathcal{A}| \, \middle| \, \mathcal{A} \text{ is a separating stratification of } \succsim\right\},
\]
and its \textit{embedded stratification number} is
\[
\mathfrak{s}^*(\succsim) := \min\left\{\mathfrak{s}(\succsim^*) \, \middle| \, \succsim^t \text{ embeds in the preorder } \succsim^*\right\}.
\]
\end{definition}

Both minima are attained: Example~\ref{example:separating_singletons} shows that the first set of cardinals is nonempty, $\succsim^t$ embeds in itself so the second is nonempty as well, and any nonempty class of cardinals has a least element. When $\succsim$ is a preorder, $\succsim^t \, = \, \succsim$ and the identity embedding gives $\mathfrak{s}^*(\succsim) \leq \mathfrak{s}(\succsim)$.

Separability of $\succsim^t$ is the statement $\mathfrak{s}^*(\succsim) \leq \aleph_0$, so Theorem~\ref{theorem:main} reads: $\succsim$ admits a Richter--Peleg representation if and only if it is strongly acyclic and $\mathfrak{s}^*(\succsim) \leq \aleph_0$. Proposition~\ref{proposition:order_density} says that $\mathfrak{s}(\succsim)$ is the least cardinality of a $\succsim$-dense subset when $\succsim$ is a complete preorder, so that for complete preorders the stratification number \emph{is} Debreu's density number. Proposition~\ref{proposition:no_embedding} gives conditions under which $\mathfrak{s}(\succsim) \leq \aleph_0$, and not merely $\mathfrak{s}^*(\succsim) \leq \aleph_0$. And Theorem~\ref{theorem:impossibility} exhibits a partial order $\succsim$ for which the two invariants come apart:
\[
\mathfrak{s}^*(\succsim) \; \leq \; \aleph_0 \; < \; \aleph_1 \; \leq \; \mathfrak{s}(\succsim).
\]

That gap is possible only because the relation is incomplete.

\begin{proposition}
\label{proposition:complete_no_gap}
Let $\succsim$ be a complete preorder with $\mathfrak{s}(\succsim)$ infinite. Then $\mathfrak{s}(\succsim) = \mathfrak{s}^*(\succsim)$.
\end{proposition}

Enlarging the domain therefore never lowers the stratification number of a complete preorder, at any infinite cardinal. The embedding clause in Definition~\ref{def:separability} is therefore not a device that a sharper argument might remove, and it binds exactly when completeness fails. The mechanism is that the increasing subsets of a complete preorder are nested---if $b \in C \setminus C'$ and $b' \in C' \setminus C$, completeness places one of them in the other's set---so any separating family of increasing sets is a chain, and a chain of a given cardinality cannot resolve more strict comparisons than a $\succsim$-dense subset of the same cardinality. Incompleteness breaks the nesting, and with it the correspondence.

\subsection{Relation to existing characterizations for preorders}
\label{sec:related_characterizations}

For preorders, the class of binary relations admitting a Richter--Peleg representation has been characterized in the prior literature through countable separating families of \textit{increasing sets}, where a set $B \subseteq X$ is increasing if $x \in B$ and $y \succsim x$ imply $y \in B$. Specifically, a preorder $\succsim$ on $X$ admits a Richter--Peleg representation if and only if there exists a countable family $\{B_n\}_{n \in \mathbb{N}}$ of increasing subsets of $X$ such that, for every pair $x \succ y$, some $B_n$ contains $x$ but not $y$. This characterization is implicit in \citet{Herden1989a} (Proposition 3.1, formulated in terms of countable families of pairs of disjoint decreasing-increasing subsets, equivalent to the increasing-set formulation via complementation) and made explicit in \citet{HackBraunGottwald2022} (Proposition 7).

Our Corollary~\ref{corollary:preorder} below recovers this characterization through Theorem~\ref{theorem:main}: the two formulations are logically equivalent for preorders. They differ in what they ask of the family. The sets $B_n$ are unconstrained relative to one another, so the condition can always be checked inside $X$; the strata of Definition~\ref{def:stratification}, by contrast, must be pairwise disjoint and rank-coherent, and Theorem~\ref{theorem:impossibility} shows that this extra structure cannot in general be realized in $X$ itself. The embedding clause is the price of it. What the price buys is Proposition~\ref{proposition:order_density}: it is precisely disjointness and rank-coherence that make a stratum behave like an indifference class, and hence make a separating stratification the right generalization of Debreu's $\succsim$-dense subset. The increasing-sets formulation, which drops that structure, connects instead to multi-utility representations, where each $B_n$ is an upper-level set $\{x : u(x) > q\}$ of a representing function.

\subsection{Corollaries}
\label{sec:corollaries}

Theorem~\ref{theorem:main} leads to two corollaries that offer simpler characterizations of Richter--Peleg representability in specific contexts.

\begin{corollary}
\label{corollary:preorder}
Let $X$ be a nonempty set and let $\succsim$ be a preorder on $X$. Then $\succsim$ admits a Richter--Peleg representation if and only if $\succsim$ is separable.
\end{corollary}

Corollary~\ref{corollary:preorder} is the stratification-based form of the characterizations discussed in Section~\ref{sec:related_characterizations}.

\begin{corollary}
\label{corollary:countable}
Let $X$ be a nonempty countable set and let $\succsim$ be a binary relation on $X$. Then $\succsim$ admits a Richter--Peleg representation if and only if $\succsim$ is strongly acyclic.
\end{corollary}

Strong acyclicity, which rules out weak-preference cycles containing a strict link, is a relatively intuitive and easily verifiable condition compared to separability. Corollary~\ref{corollary:countable} shows that when $X$ is countable, this condition alone is both necessary and sufficient for Richter--Peleg representability. This is close to classical territory: by Proposition~\ref{proposition:equivalence}, strong acyclicity is Suzumura consistency, which \citet[Theorem~3]{Suzumura1976} shows to be equivalent to the existence of a complete preorder extending $\succsim$, and every complete preorder on a countable set is representable.

\section{Applications}
\label{sec:application}

\subsection{Optimality and utility maximization}
\label{sec:optimization}

Identifying the class of preorders with a Richter--Peleg representation is important because optimal behavior according to such preorders can be characterized by scalar optimization.

In the theory of choice under risk, it has long been known that every lottery maximal for a DMO preference over a polytope\footnote{\textit{DMO} stands for Dubra--Maccheroni--Ok; these preferences and Aumann utilities were introduced in Example~\ref{example:mixture}. A \textit{polytope} is the convex hull of a finite set of points; \citet{Aumann1962} uses the term ``convex polyhedron''.} maximizes the expectation of some Aumann utility for the preference. This is the ``only if'' direction of \citet[Theorem~B]{Aumann1962}; that direction is false as originally stated and was re-established by \citet{Aumann1964} under a strengthened Archimedean axiom, which the continuity of a DMO preference supplies. For the generality needed here, see \citet{Evren2014} and \citet{Gorno2017}.

For a binary relation $\succsim$ on $X$ and $A \subseteq X$, write
\[
\mathcal{M}\left(\succsim, A\right) := \left\{x \in A \, \middle| \, \forall y \in A: y \succsim x \implies x \succsim y \right\}
\]
for the set of $\succsim$-maximal elements of $A$. \citet{White1980} showed that similar results hold for general preorders whenever a Richter--Peleg representation exists. His Theorem~1 states that, if $\succsim$ is a preorder that admits a Richter--Peleg representation, then
\[
\mathcal{M}\left(\succsim, X\right) = \bigcup_{u \in \mathcal{U}} \argmax_{x \in X} u(x),
\]
where $\mathcal{U}$ denotes the set of all Richter--Peleg representations of $\succsim$.\footnote{White's notation is the reverse of ours: his $M(X,V)$ is the union of the $\argmax$ sets and his $\mathcal{C}(X,R)$ is our $\mathcal{M}(\succsim,X)$. He works with the bounded members of $\mathcal{U}$, which is immaterial here, since composing a Richter--Peleg representation with a strictly increasing bounded transformation preserves both the property and the set of maximizers.} His Theorem~2 sharpens this to a single representation, but only for complete preorders. Proposition~\ref{proposition:optimality1} below extends both statements: it holds for an arbitrary subset $A \subseteq X$, it requires no completeness, and it delivers a single Richter--Peleg representation. The price is that the representation depends on $A$, and Example~\ref{example:A_dependence} shows that this cannot be avoided---which is precisely why completeness was needed in White's Theorem~2.

\citet{White1980} does not establish that $\mathcal{U} \neq \emptyset$; he assumes it, extending $\succsim$ to a complete preorder by Szpilrajn's theorem and then invoking a countable denseness hypothesis in order to represent the extension. Theorem~\ref{theorem:main} characterizes exactly when that hypothesis can be met.

\begin{proposition}
\label{proposition:optimality1}
Suppose the preorder $\succsim$ admits a Richter--Peleg representation. Then, for every $A \subseteq X$, there exists a Richter--Peleg representation of $\succsim$, $u:X\to \mathbb{R}$, such that
\[
\mathcal{M}\left(\succsim, A\right) = \argmax_{x \in A} u(x).
\]
\end{proposition}

The representation in Proposition~\ref{proposition:optimality1} depends on $A$, and this dependence is unavoidable.

\begin{example} \label{example:A_dependence} Let $X=\{a,b,c\}$ and $\succsim \hspace{2pt} := \left\{(a,a),(b,b),(c,c),(a,c)\right\}$, so that $a \succ c$ and $b$ is incomparable to both $a$ and $c$. The relation $\succsim$ is a partial order and admits the Richter--Peleg representation $v(a) = 1$, $v(b) = v(c) = 0$. Suppose some single $u$ satisfied $\mathcal{M}\left(\succsim, A\right) = \argmax_{x \in A} u(x)$ for every $A \subseteq X$. Taking $A = \{a,b\}$, whose two elements are both maximal, gives $u(a) = u(b)$; taking $A = \{b,c\}$ gives $u(b) = u(c)$; hence $u(a) = u(c)$, contradicting $a \succ c$.
\end{example}

A second result shows that every alternative maximizes some Richter--Peleg representation over every set in which the alternative is optimal.

\begin{proposition}
\label{proposition:optimality2}
Suppose the preorder $\succsim$ admits a Richter--Peleg representation. Then, for every $x \in X$, there exists a Richter--Peleg representation of $\succsim$, $u : X \to \mathbb{R}$, such that
\[
x \in \argmax_{y \in A} u(y)
\]
for every $A \subseteq X$ satisfying $x \in \mathcal{M}\left(\succsim, A\right)$.
\end{proposition}

Together, Propositions~\ref{proposition:optimality1} and \ref{proposition:optimality2} show that the maximization of preorders admitting a Richter--Peleg representation can always be reduced to scalar optimization. This result does not extend to strongly acyclic binary relations that fail transitivity:

\begin{example}
Let $X=\{a,b,c\}$ and $\succsim \hspace{2pt} := \left\{(a,a),(b,b),(c,c),(a,b),(b,c),(c,b)\right\}$. The relation $\succsim$ is not transitive but admits the Richter--Peleg representation $u: X \to \mathbb{R}$ given by $u(a)=1$ and $u(b)=u(c)=0$. By Lemma~\ref{lemma:strong_acyclicity}, $\succsim$ is strongly acyclic. The set of maximal elements is $\mathcal{M}\left(\succsim,\{a,c\}\right) = \{a,c\}$. Nevertheless, $c \notin \argmax_{x \in \{a,c\}} v(x)$ for any Richter--Peleg representation $v: X \to \mathbb{R}$ of $\succsim$, so the conclusions of Propositions~\ref{proposition:optimality1} and \ref{proposition:optimality2} fail.
\end{example}

\subsection{An application to intergenerational equity}
\label{sec:equity}

Richter--Peleg representability also organizes a known impossibility result in welfare economics. Let $X := [0,1]^{\mathbb{N}}$ be the set of infinite utility streams, let $\geq$ denote the componentwise order, and let $\Pi$ be the set of permutations of $\mathbb{N}$ moving only finitely many coordinates. The \textit{Suppes--Sen grading preorder} is defined by $x \succsim_{S} y$ if and only if $\pi x \geq y$ for some $\pi \in \Pi$, where $(\pi x)_{n} := x_{\pi(n)}$; it is the canonical formalization of a social ranking that respects both efficiency and the equal treatment of generations. A \textit{social welfare function} $W : X \to \mathbb{R}$ satisfies \textit{strong Pareto} if $x \geq y$ and $x \neq y$ imply $W(x) > W(y)$ for all $x,y \in X$, and \textit{finite anonymity} if $W(\pi x) = W(x)$ for all $\pi \in \Pi$ and $x \in X$.

\begin{proposition}
\label{proposition:suppes_sen}
A social welfare function satisfies strong Pareto and finite anonymity if and only if it is a Richter--Peleg representation of $\succsim_{S}$.
\end{proposition}

Since $\succsim_{S}$ is a preorder, it is strongly acyclic and equal to its own transitive closure, so Theorem~\ref{theorem:main} reduces to separability alone. The impossibility theorem of \citet{BasuMitra2003}---no social welfare function on $X$ satisfies strong Pareto and finite anonymity---therefore says exactly that $\succsim_{S}$ is not separable.\footnote{We use the theorem only in the form stated here, for one-period utilities in $[0,1]$; see \citet{BasuMitra2003} for the general statement.} Proposition~\ref{proposition:order_density}(i) then yields a concrete diagnosis: no countable set of utility streams can be interposed in every strict Suppes--Sen comparison, so the obstruction is Debreu's classical one, arising here in a relation that is very far from complete.\footnote{That equitable social welfare relations admit no Richter--Peleg representation is established directly by \citet{BanerjeeDubey2010}, who also rule out countable multi-utility representations. What Theorem~\ref{theorem:main} adds is the identification of the property whose failure is responsible, namely separability.}

The failure is one of representation rather than of ranking. \citet{Svensson1980} shows, via the Szpilrajn extension theorem, that Paretian and anonymous social welfare \textit{orderings} exist; Proposition~\ref{proposition:equivalence} and \citet[Theorem~3]{Suzumura1976} say the same thing for any preorder, since every preorder is strongly acyclic and hence extends to a complete preorder. What the welfare literature's impossibility adds is that no such extension is representable, and Theorem~\ref{theorem:main} locates the reason in the failure of separability in the sense of Definition~\ref{def:separability}. Nor is the failure inherent to the equity requirement: on a finite horizon, total utility $x \mapsto \sum_{i \leq n} x_{i}$ is anonymous and strictly increasing in the componentwise order, hence a Richter--Peleg representation of the finite-horizon Suppes--Sen preorder.

\section{Discussion}
\label{sec:discussion}

The conditions of Theorem~\ref{theorem:main} bear differently on different classes of binary relations of interest in economics. For preorders, separability is the only binding condition, and Corollary~\ref{corollary:preorder} is directly applicable. Many incomplete preferences arising in economic theory are preorders: Pareto dominance on consumption bundles, stochastic dominance orderings on lotteries, and the canonical preorders associated with multi-criterion or expected multi-utility models \citep{Ok2002,Dubra2004,EvrenOk2011}. Necessary and sufficient conditions for Richter--Peleg representability in these settings are also available through the increasing-sets formulation of \citet{Herden1989a} and \citet{HackBraunGottwald2022}; the distinguishing feature of our characterization is that it specializes, under completeness, to Debreu's classical density condition (Proposition~\ref{proposition:order_density}), making the parallel with the classical representation theorem explicit.

For binary relations that fail to be transitive, the reach of Theorem~\ref{theorem:main} is narrower. Strong acyclicity rules out a broad class of behaviorally motivated non-transitive preferences, including complete semiorders with intransitive indifference \citep{Luce1956}\footnote{Luce's axioms impose total comparability (Axiom~S1); later work relaxes this, studying semiorders as a special class of possibly incomplete \textit{interval orders} \citep{Fishburn1970}. We retain completeness here because it is essential to the mechanism: incomparability, unlike a resolved strict comparison, need not close the weak-preference cycle that strong acyclicity forbids.} and similar models in which a strict preference between two alternatives may coexist with a chain of pairwise indifferences linking them.\footnote{In a semiorder on $\mathbb{R}$ with a constant discrimination threshold $\delta > 0$---the special case in which the upper and lower just-noticeable-difference functions of \citet[Corollary to Theorem~2]{Luce1956} coincide and do not vary with the point of evaluation---one typically has chains $x_1 \sim x_2 \sim \ldots \sim x_n$ with $x_n \succ x_1$, which violates strong acyclicity. This is in fact a general feature of complete semiorders, not an artifact of the threshold representation: whenever $x \sim y \sim z$ but $x \succ z$ (the case $z \succ x$ is symmetric), the three-term cycle $x \succsim z \succsim y \succsim x$ already violates strong acyclicity, since it does not collapse to $x \sim z \sim y \sim x$.} Such preferences therefore do not admit a Richter--Peleg representation (Lemma~\ref{lemma:strong_acyclicity}). The affirmative content of Theorem~\ref{theorem:main} for non-transitive relations lies instead in settings where the analyst observes only a subset of the true preference comparisons. The transitive closure of the observed relation captures the implications of the recorded comparisons, while strong acyclicity ensures these implications are consistent with the underlying strict preferences---a perspective with antecedents in the revealed-preference tradition initiated by \citet{Richter1966} and developed by \citet{Suzumura1976}, where conditions ruling out cycles of weak and strict comparison play a central role.

Figure~\ref{fig:venn} summarizes how strong acyclicity and separability, jointly equivalent to Richter--Peleg representability by Theorem~\ref{theorem:main}, relate to completeness and transitivity. Completeness and Richter--Peleg representability jointly imply transitivity,\footnote{If $u$ is a Richter--Peleg representation and $x \succsim y \succsim z$, isotonicity gives $u(x) \geq u(z)$; were $x \not\succsim z$, completeness would give $z \succ x$ and hence $u(z) > u(x)$.} so of the eight configurations of the three properties in Figure~\ref{fig:venn}, only region~(ii) is empty; the remaining seven are populated by examples already given in this paper. Region~(i), where all three properties hold, is the classical case of a complete, transitive relation admitting an ordinary utility representation. Region~(iii), completeness alone, is exemplified by the three-element cycle of Example~\ref{example:strict_cycle}, which fails both transitivity and strong acyclicity. Region~(iv), Richter--Peleg representability alone, is exemplified by Example~\ref{example:transitivity}. Region~(vi), completeness and transitivity without Richter--Peleg representability, is exemplified by the lexicographic order of Example~\ref{example:lex}. Region~(vii), Richter--Peleg representability and transitivity without completeness, is populated by any finite partial order that is not total, which is representable by Corollary~\ref{corollary:countable} since preorders are automatically strongly acyclic. Region~(v), transitivity alone, is exemplified by adjoining to the lexicographic order of Example~\ref{example:lex} a single further alternative incomparable to everything else: the enlarged relation remains transitive, since the new alternative bears no relation to the rest, yet fails completeness through the new incomparability and Richter--Peleg representability because any representation would restrict to one of Example~\ref{example:lex}. The same device, adjoining an incomparable alternative to Example~\ref{example:strict_cycle}, populates region~(viii), where none of the three properties holds.

\begin{figure}[htbp]
\centering
\begin{tikzpicture}[scale=0.95]
\begin{scope}[fill opacity=0.20, draw opacity=0.85]
\draw[completeblue, fill=completeblue, solid, very thick] (-1.8,0.75) ellipse (2.7 and 1.65);
\draw[rpred, fill=rpred, dashed, very thick] (1.8,0.75) ellipse (2.7 and 1.65);
\draw[transgreen, fill=transgreen, dotted, very thick] (0,-0.85) ellipse (2.4 and 1.9);
\end{scope}
\node[completeblue] at (-2.3,1.8) {\small Complete};
\node[rpred,align=center] at (2.2,1.6) {\small Richter--Peleg\\[-2pt]\small representable};
\node[transgreen] at (0,-2.3) {\small Transitive};
\node at (0,0.55) {(i)};
\node at (0,1.45) {(ii)};
\node at (-3.05,0.75) {(iii)};
\node at (3.15,0.65) {(iv)};
\node at (0,-1.6) {(v)};
\node at (-1.3,-0.45) {(vi)};
\node at (1.3,-0.45) {(vii)};
\node at (-3.85,-1.85) {(viii)};
\end{tikzpicture}
\caption{The relationship between completeness, transitivity, and Richter--Peleg representability. Ellipse boundaries are distinguished by line style (solid: complete; dashed: Richter--Peleg representable; dotted: transitive), and regions are numbered for reference in the text. Region~(ii), where a relation is complete and Richter--Peleg representable but not transitive, is empty; every other region is nonempty.}
\label{fig:venn}
\end{figure}

When $X$ carries a topology, a natural further question is whether the Richter--Peleg representation produced by Theorem~\ref{theorem:main} can be chosen continuous, or at least upper or lower semicontinuous. For complete preorders, this is the classical question addressed by \citet{Eilenberg1941} and \citet{Debreu1964}, with \citet{Peleg1970} treating partial orders and modern treatments in \citet{Herden1989a,Herden1989b}, \citet{BridgesMehta1995}, and \citet{BosiHerden2012}. For incomplete preorders, sufficient conditions for the existence of continuous Richter--Peleg representations have been developed by \citet{Bosi2010} and refined by \citet{Alcantud2016}, among others. A general topological refinement of Theorem~\ref{theorem:main}---identifying conditions on $(X, \succsim)$ under which the Richter--Peleg representation can be chosen with prescribed continuity properties---is beyond the scope of the present paper but constitutes a natural direction for future work.

\section*{Acknowledgments}
We thank Guilherme Araújo Lima for excellent research assistance. We are also grateful to participants at various seminars for their helpful comments. This study was financed in part by the Coordenação de Aperfeiçoamento de Pessoal de N\'{i}vel Superior - Brasil (CAPES) - Finance Code 001. The authors also gratefully acknowledge the financial assistance of the Brazilian National Council for Scientific and Technological Development (CNPq).

\appendix

\section{Proofs}
\label{appendix:proofs}

\begin{proof}[Proof of Lemma~\ref{lemma:strong_acyclicity}]
Let $u: X \to \mathbb{R}$ be a Richter--Peleg representation of $\succsim$. Take alternatives $x_1,x_2,\ldots,x_n \in X$ such that $x_1 \succsim x_2 \succsim \ldots \succsim x_n \succsim x_1$. Then, $u(x_1) \geq u(x_2) \geq \ldots \geq u(x_n) \geq u(x_1)$. It follows that $u(x_1) = u(x_2) = \ldots = u(x_n)$. Now suppose, seeking a contradiction, that $x_1 \sim x_2 \sim \ldots \sim x_n \sim x_1$ fails to hold. Hence, at least one of the indifferences does not hold. Because we can always shift the cycle, we can assume without loss of generality that $x_1 \not\sim x_2$. Hence, $x_1 \succ x_2$. Since $u$ is a Richter--Peleg representation of $\succsim$, we must have $u(x_1) > u(x_2)$, a contradiction.
\end{proof}

\begin{proof}[Proof of Proposition~\ref{proposition:equivalence}]
We prove the chain of implications $(1)\Rightarrow(4)\Rightarrow(2)\Rightarrow(3)\Rightarrow(1)$.

$(1)\Rightarrow(4)$. Suppose $\succsim$ is strongly acyclic and $x \succ y$. Then $x \succsim y$, so $x \succsim^t y$. Suppose, for contradiction, $y \succsim^t x$. Since neither $y=x$ nor $y \succsim x$ can hold, there must be a chain $y \succsim z_1 \succsim \ldots \succsim z_k \succsim x$, which combined with $x \succsim y$ yields $y \succsim z_1 \succsim \ldots \succsim z_k \succsim x \succsim y$. By strong acyclicity, $x \sim y$, contradicting $x \succ y$. Hence $y \not\succsim^t x$, so $x \succ^t y$.

$(4)\Rightarrow(2)$. Suppose $\succsim^t$ is an extension of $\succsim$, and suppose for contradiction that there is a chain $x_1 \succ x_2 \succsim x_3 \succsim \ldots \succsim x_n \succsim x_1$. By transitivity of $\succsim^t$, $x_2 \succsim^t x_1$. By assumption, $x_1 \succ x_2$ implies $x_1 \succ^t x_2$, so $x_2 \not\succsim^t x_1$, a contradiction.

$(2)\Rightarrow(3)$. Suppose $\succsim$ is Suzumura consistent, and suppose for contradiction that $(\succsim, \succ)$ is not pair-acyclic: there is a chain $x_1 \succsim x_2 \succsim \ldots \succsim x_n \succ x_1$. Cyclically shift to obtain $x_n \succ x_1 \succsim x_2 \succsim \ldots \succsim x_{n-1} \succsim x_n$, a chain witnessing the failure of Suzumura consistency, a contradiction.

$(3)\Rightarrow(1)$. Suppose $(\succsim, \succ)$ is pair-acyclic, and suppose for contradiction that $\succsim$ is not strongly acyclic. Then there exist $x_1, \ldots, x_n \in X$ with $x_1 \succsim x_2 \succsim \ldots \succsim x_n \succsim x_1$ but $x_i \not\sim x_{i+1}$ for some $i$ (indices mod $n$, with $x_{n+1} := x_1$). Then $x_i \succ x_{i+1}$. Cyclically shift so that the strict link comes last: $x_{i+1} \succsim x_{i+2} \succsim \ldots \succsim x_n \succsim x_1 \succsim \ldots \succsim x_i$ followed by $x_i \succ x_{i+1}$. This is a chain witnessing the failure of pair-acyclicity, a contradiction.
\end{proof}

\begin{proof}[Proof of Lemma~\ref{lemma:transitivity_embedding}]
Let $\succsim$ be a binary relation that can be embedded in a preorder $\succsim^*$, with embedding function $f:X\to X^*$. Since $\succsim^*$ is reflexive, $f(x) \succsim^* f(x)$ for all $x \in X$. By definition of embedding, $x \succsim x$ for all $x \in X$, so $\succsim$ is reflexive. Now, let $x, y, z \in X$ be such that $x \succsim y \succsim z$. Then, $f(x) \succsim^* f(y) \succsim^* f(z)$. Since $\succsim^*$ is transitive, $f(x) \succsim^* f(z)$. By definition of embedding $x \succsim z$. Hence, $\succsim$ is transitive. We conclude $\succsim$ is a preorder.
\end{proof}

To prove Theorem~\ref{theorem:main}, we establish a series of auxiliary lemmas.

\begin{lemma}
\label{lemma:reflexivity}
Let $\succsim$ be a binary relation on a set $X$ and define $\succsim' \hspace{2pt}:= \hspace{2pt} \succsim \hspace{2pt} \cup \left\{(x,x) \, \middle| \, x \in X\right\}$. Then $\succsim$ and $\succsim'$ have the same set of Richter--Peleg representations.
\end{lemma}

\begin{proof}
Let $u:X\to\mathbb{R}$ be a Richter--Peleg representation of $\succsim$. Then the only additional restriction for $u$ to be a Richter--Peleg representation of $\succsim'$ is that $u(x)=u(x)$, which obviously holds. Conversely, if $u$ is a Richter--Peleg representation of $\succsim'$, then $u$ is also a Richter--Peleg representation of $\succsim$ because $x \succ y$ implies $x \succ' y$ and $x \sim y$ implies $x \sim' y$ for all $x,y \in X$.
\end{proof}

\begin{lemma}
\label{lemma:transitive_closure_1}
If the transitive closure of a strongly acyclic binary relation $\succsim$ admits a Richter--Peleg representation, then $\succsim$ admits a Richter--Peleg representation.
\end{lemma}

\begin{proof}
Let $\succsim^t$ denote the transitive closure of $\succsim$ and let $u: X \to \mathbb{R}$ be a Richter--Peleg representation of $\succsim^t$. Suppose $x \succsim y$. Then $x \succsim^t y$, so $u(x) \geq u(y)$. Now suppose $x \succ y$. By Proposition~\ref{proposition:equivalence}, $x \succ^t y$, so $u(x) > u(y)$. We conclude $u$ is a Richter--Peleg representation of $\succsim$.
\end{proof}

\begin{lemma}
\label{lemma:sufficiency}
If a preorder $\succsim$ possesses a countable separating stratification, then $\succsim$ admits a Richter--Peleg representation.
\end{lemma}

\begin{proof}
Let $\mathcal{A} = \left\{A_i \, \middle| \, i \in I\right\}$ be a countable separating stratification of $\succsim$, where $I \subseteq \mathbb{N}$ without loss of generality. For each $i \in I$, consider the weak and strict increasing hulls $\uphull{A_i}$ and $\Uphull{A_i}$ from Section~\ref{sec:stratifications}. Since $\succsim$ is transitive, both are increasing: if $x \in \uphull{A_i}$ (say $x \succsim z$ with $z \in A_i$) and $w \succsim x$, then $w \succsim z$, so $w \in \uphull{A_i}$; and if $x \in \Uphull{A_i}$ (say $x \succ z$) and $w \succsim x$, then $w \succ z$ by \eqref{eq:strict_composition}, so $w \in \Uphull{A_i}$. Define $u: X \to \mathbb{R}$ by
\[
u(x) := \sum_{i \in I} 2^{-i} \left[\mathbf{1}_{\uphull{A_i}}(x) + \mathbf{1}_{\Uphull{A_i}}(x)\right],
\]
where $\mathbf{1}_{(\cdot)}$ denotes the indicator function. The series converges because each bracket lies in $\{0, 1, 2\}$ and $\sum_{i \in I} 2^{-i} \leq \sum_{i \in \mathbb{N}} 2^{-i} = 1$.

We verify that $u$ is a Richter--Peleg representation of $\succsim$. Take $x, y \in X$.

Suppose $x \sim y$. Since both hulls are increasing and $x \succsim y$ and $y \succsim x$, each indicator takes the same value at $x$ and $y$; hence $u(x) = u(y)$.

Suppose $x \succ y$. Then $x \succsim y$, and since both hulls are increasing, $\mathbf{1}_{\uphull{A_i}}(x) \geq \mathbf{1}_{\uphull{A_i}}(y)$ and $\mathbf{1}_{\Uphull{A_i}}(x) \geq \mathbf{1}_{\Uphull{A_i}}(y)$ for every $i$; thus $u(x) \geq u(y)$. It remains to exhibit a strict gain. Because $\mathcal{A}$ is separating, there is $A_j \in \mathcal{A}$ with $x \succsim z_0$ for some $z_0 \in A_j$, so that $x \in \uphull{A_j}$, and with $y \not\succ z'$ for every $z' \in A_j$, so that $y \notin \Uphull{A_j}$. We consider the $j$-th term in two cases, according to whether $y \in \uphull{A_j}$.

If $y \notin \uphull{A_j}$, then, since $x \in \uphull{A_j}$, we get $\mathbf{1}_{\uphull{A_j}}(x) = 1 > 0 = \mathbf{1}_{\uphull{A_j}}(y)$, a strict gain.

If $y \in \uphull{A_j}$, say $y \succsim z_1$ for some $z_1 \in A_j$, then $x \succ y \succsim z_1$ gives $x \succ z_1$ by \eqref{eq:strict_composition}, placing $x \in \Uphull{A_j}$. Since $y \notin \Uphull{A_j}$, this yields the strict gain $\mathbf{1}_{\Uphull{A_j}}(x) = 1 > 0 = \mathbf{1}_{\Uphull{A_j}}(y)$.

In either case the $j$-th term strictly increases while no term decreases, and we conclude $u(x) > u(y)$.
\end{proof}

Note that the proof uses only countability and the separating property; neither the pairwise disjointness nor the rank-coherence required of a stratification enters. Both are nevertheless delivered by the construction in Lemma~\ref{lemma:preorder_case}, so the stronger object is obtained free of charge in the direction where it is harder to come by.

\begin{lemma}
\label{lemma:transitive_closure_2}
If a binary relation $\succsim$ admits a Richter--Peleg representation, then the transitive closure of $\succsim$ admits a Richter--Peleg representation.
\end{lemma}

\begin{proof}
Let $u: X \to \mathbb{R}$ be a Richter--Peleg representation of $\succsim$ and let $\succsim^t$ be the transitive closure of $\succsim$. Suppose $x \succsim^t y$. Then $x=y$, $x \succsim y$, or $x \succsim x_1 \succsim \ldots \succsim x_n \succsim y$. Each of these three cases implies $u(x) \geq u(y)$. Now suppose $x \succ^t y$. We cannot have $x = y$. If $x \succ y$, then $u(x) > u(y)$. If $x \sim x_1 \sim \ldots \sim x_n \sim y$ we would have $x \sim^t y$, which is impossible. Then, we must have $x_k \succ x_{k+1}$ for some $k \in \{0,\ldots,n\}$, where $x_0 = x$ and $x_{n+1} = y$. It follows that $u(x) > u(y)$. We conclude that $u$ is a Richter--Peleg representation of $\succsim^t$.
\end{proof}

The following result is key to proving the necessity part of Theorem~\ref{theorem:main}.

\begin{lemma}
\label{lemma:preorder_case}
If a preorder $\succsim$ admits a Richter--Peleg representation, then $\succsim$ can be embedded in a preorder that possesses a countable separating stratification.
\end{lemma}

\begin{proof}
Let $\succsim$ be a preorder on $X$ with a Richter--Peleg representation $u: X \to \mathbb{R}$. We construct the embedding---which we call the \textit{shadow completion} of $\succsim$---by adjoining to $X$, for each alternative and each rational level strictly below its utility, an auxiliary \textit{shadow} of that alternative at that level.\footnote{The carrier $X^*$ resembles in form the formal-ball model $\mathbf{B}X = X \times \mathbb{R}_+$ of \citet{EdalatHeckmann1998} (see also \citealp{GierzEtAl2003}, Example V-6.8), but the two constructions serve opposite ends. In the formal-ball model the second coordinate is an approximation radius, and the order $(x,r) \le (y,s) \iff d(x,y) \le r - s$ is designed to make $\mathbf{B}X$ a continuous poset that recovers $X$ as its space of maximal points. Here the second coordinate is a utility threshold, the order is inherited from $\succsim$, and the resulting preorder is deliberately neither directed nor equipped with a maximal frontier; its strata instead form a countable family separating the strict comparisons of an incomplete preorder.}

Formally, let
\[
S := \left\{(x,q) \, \middle| \, x \in X,\; q \in \mathbb{Q},\; q < u(x)\right\},
\qquad
X^* := X \sqcup S,
\]
the disjoint union of $X$ with the set $S$ of shadows; we call the elements of $X$ \textit{originals}. Define $u^* : X^* \to \mathbb{R}$ by $u^*(x) := u(x)$ for an original $x$ and $u^*\big((x,q)\big) := q$ for a shadow $(x,q)$. Define a binary relation $\succsim^*$ on $X^*$ by stipulating, for all $x,y \in X$ and all $(x,q),(y,q') \in S$,
\begin{align*}
x \succsim^* y \quad&\Longleftrightarrow\quad x \succsim y, \\
x \succsim^* (y,q) \quad&\Longleftrightarrow\quad x \succsim y, \\
(x,q) \succsim^* y \quad&\text{never holds}, \\
(x,q) \succsim^* (y,q') \quad&\Longleftrightarrow\quad \big(q > q' \text{ and } x \succsim y\big) \text{ or } \big(q = q' \text{ and } x = y\big).
\end{align*}
The decisive feature is the last clause: shadows of \emph{distinct} originals at a \emph{common} level are $\succsim^*$-incomparable. We record an observation used repeatedly: a shadow is never $\succsim^*$-above an original, so every $\succsim^*$-relation is of type original--original, original--shadow, or shadow--shadow. Figure~\ref{fig:shadow} in Section~\ref{sec:impossibility} illustrates the construction.

\emph{$\succsim^*$ is a preorder.} Reflexivity is immediate ($x \succsim x$; and $q=q$, $x=x$). For transitivity, suppose $a \succsim^* b \succsim^* c$. By the observation, neither $b$ nor $c$ can be an original once $a$ is a shadow, and $c$ cannot be an original once $b$ is a shadow; hence only the type patterns original--original--original, original--original--shadow, original--shadow--shadow, and shadow--shadow--shadow can occur. In the first two type patterns, the relations reduce via the first two clauses to $\succsim$ on the originals involved, and transitivity of $\succsim$ gives $a \succsim^* c$. In the original--shadow--shadow case, the second and fourth clauses give $a \succsim^* c$. Finally, in the shadow--shadow--shadow case, write $a=(x,q)$, $b=(y,p)$, $c=(z,r)$; each premise is $(q>p \wedge x \succsim y)\vee(q=p\wedge x=y)$ and $(p>r \wedge y \succsim z)\vee(p=r\wedge y=z)$. If both premises use the equality disjunct, then $q=p=r$ and $x=y=z$, so $a \succsim^* c$; otherwise at least one level inequality is strict, giving $q > r$, while $x \succsim z$ follows from transitivity of $\succsim$ (combining $x \succsim y$ or $x=y$ with $y \succsim z$ or $y=z$). Hence $a \succsim^* c$.

\emph{$\succsim^*$ embeds $\succsim$.} Let $\iota : X \to X^*$ be the inclusion of originals. For originals $x,y$, the pair $(\iota(x),\iota(y))$ is governed exclusively by the first clause, so $\iota(x) \succsim^* \iota(y)$ if and only if $x \succsim y$. Thus $\succsim^*$ is an embedding of $\succsim$ in the sense of Definition~\ref{def:embedding}.

\emph{A countable separating stratification.} For $r \in \mathbb{Q}$, let
\[
A_r := \left\{(x,r) \, \middle| \, x \in X,\; r < u(x)\right\}
\]
be the set of all shadows at level $r$, and put $\mathcal{A} := \left\{A_r \, \middle| \, r \in \mathbb{Q},\ A_r \neq \emptyset\right\}$. Since $\mathcal{A}$ is indexed by $\mathbb{Q}$, it is countable. For given $r \in \mathbb{Q}$, we compute the increasing hull $\uphull{A_r} = \left\{z \in X^* \, \middle| \, z \succsim^* a \text{ for some } a \in A_r\right\}$ formed with respect to $\succsim^*$. For an original $z$,
\[
z \in \uphull{A_r} \iff r < u^*(z),
\]
since $z \succsim^* (x,r)$ reduces to $z \succsim x$, which holds for the witness $x = z$ exactly when $(z,r) \in A_r$, i.e.\ when $r < u(z)$, and conversely $z \succsim x$ with $r < u(x)$ gives $r < u(x) \leq u(z)$ by isotonicity of $u$. For a shadow $z = (a,p)$,
\[
z \in \uphull{A_r} \iff r \leq p,
\]
since $(a,p) \succsim^* (x,r)$ requires $p > r$ (with $a \succsim x$) or $p = r$ (with $a = x$), and conversely $r \leq p$ is witnessed by $x = a$ (using $r \leq p < u(a)$, so $(a,r) \in A_r$).

Each $A_r \in \mathcal{A}$ is nonempty by construction, and the strata are pairwise disjoint: every element of $A_r$ has second coordinate $r$, so a shadow belongs to exactly one $A_r$. Moreover, $\mathcal{A}$ satisfies the rank-coherence condition of a stratification of $\succsim^*$: for shadows, $(x,q) \succ^* (y,q')$ holds only if $q > q'$ (the equality disjunct forces equality of the two shadows, hence not a strict relation). Thus given $A_q, A_{q'} \in \mathcal{A}$ with $a, b \in A_q$ and $a', b' \in A_{q'}$, the premise $a \succ^* a'$ forces $q > q'$, which precludes $b' \succ^* b$ (that would force $q' > q$); in particular, taking $A_q = A_{q'}$ shows no stratum contains a strictly ranked pair.

Finally, $\mathcal{A}$ is separating. Let $z,w \in X^*$ with $z \succ^* w$. We claim $u^*(w) < u^*(z)$. Indeed, if $z,w$ are originals this is the strict part of the Richter--Peleg property of $u$; if $z$ is an original and $w = (b,s)$ a shadow, then $z \succsim^* w$ gives $z \succsim b$, so $u^*(z) = u(z) \geq u(b) > s = u^*(w)$; and if $z = (c,t)$, $w = (b,s)$ are shadows, then $z \succ^* w$ uses the first disjunct, so $t > s$, i.e.\ $u^*(z) > u^*(w)$. (The remaining type, $z$ a shadow and $w$ an original, cannot yield $z \succ^* w$.) By density of $\mathbb{Q}$, choose $r \in \mathbb{Q}$ with $u^*(w) < r < u^*(z)$. By the hull computations, $z \in \uphull{A_r}$ (since $r < u^*(z)$ covers both the original and shadow cases) and $w \notin \uphull{A_r}$ (since $r > u^*(w)$ fails both $r < u^*(w)$ and $r \leq u^*(w)$); moreover $A_r \neq \emptyset$ because $z \in \uphull{A_r}$ exhibits a member of $A_r$. Equivalently, there is $a \in A_r$ with $z \succsim^* a$, while $w \not\succsim^* a'$---hence a fortiori $w \not\succ^* a'$---for every $a' \in A_r$. Thus $A_r \in \mathcal{A}$ separates the pair in the sense of Definition~\ref{def:separating_stratification}, and $\mathcal{A}$ is separating.

We conclude that $\mathcal{A}$ is a countable separating stratification of the preorder $\succsim^*$, into which $\succsim$ embeds.
\end{proof}

We are now ready to present the proof of Theorem~\ref{theorem:main}.

\begin{proof}[Proof of Theorem~\ref{theorem:main}]
By Lemma~\ref{lemma:reflexivity}, there is no loss of generality in assuming $\succsim$ is reflexive, which we do throughout. Let $\succsim^t$ denote its transitive closure, which is a preorder.

\emph{Sufficiency.} Suppose $\succsim$ is strongly acyclic and $\succsim^t$ is separable. By definition, there is an embedding $f$ of $\succsim^t$ into a preorder $\succsim^*$ that possesses a countable separating stratification. By Lemma~\ref{lemma:sufficiency}, $\succsim^*$ admits a Richter--Peleg representation $u^*$; then $u^* \circ f$ is a Richter--Peleg representation of $\succsim^t$.\footnote{If $x \sim^t y$, then $f(x) \sim^* f(y)$, so $u^*(f(x)) = u^*(f(y))$; if $x \succ^t y$, then $f(x) \succ^* f(y)$, so $u^*(f(x)) > u^*(f(y))$.} Since $\succsim$ is strongly acyclic, Lemma~\ref{lemma:transitive_closure_1} yields a Richter--Peleg representation of $\succsim$.

\emph{Necessity.} Suppose $\succsim$ admits a Richter--Peleg representation. By Lemma~\ref{lemma:strong_acyclicity}, $\succsim$ is strongly acyclic. By Lemma~\ref{lemma:transitive_closure_2}, the preorder $\succsim^t$ also admits a Richter--Peleg representation. By Lemma~\ref{lemma:preorder_case}, $\succsim^t$ can therefore be embedded in a preorder possessing a countable separating stratification; that is, $\succsim^t$ is separable.
\end{proof}

\begin{proof}[Proof of Proposition~\ref{proposition:order_density}]
\textbf{Part (i).} Let $Z$ be a $\succsim$-dense subset of $X$ and set $\mathcal{A} := \left\{\{z\} \, \middle| \, z \in Z\right\}$. The singletons are nonempty and pairwise disjoint, and for $A = \{z\}$ and $A' = \{z'\}$ the condition in Definition~\ref{def:stratification} reduces to ``$z \succ z'$ implies $z' \not\succ z$'', which is asymmetry of $\succ$ and therefore holds for every binary relation. Hence $\mathcal{A}$ is a stratification, of cardinality $|Z|$. Note that no property of $\succsim$ has been used. To verify separation, take $x, y \in X$ with $x \succ y$. By $\succsim$-density there is $z \in Z$ with $x \succsim z \succsim y$. Then $x \succsim z$, and $z \succsim y$ gives $\neg(y \succ z)$; so the singleton $\{z\}$ satisfies $x \succsim z$ and $y \not\succ z'$ for its unique element $z' = z$. Hence $\{z\}$ separates $(x,y)$, and $\mathcal{A}$ is separating.

\textbf{Part (ii).} Let $\mathcal{A}$ be a separating stratification of $\succsim$, choose an arbitrary representative $z_A \in A$ for each $A \in \mathcal{A}$, and set $Z := \left\{z_A \, \middle| \, A \in \mathcal{A}\right\}$. Since the strata of $\mathcal{A}$ are pairwise disjoint, distinct strata have distinct representatives, so $|Z| = |\mathcal{A}|$.

We first observe that any two elements $w, w'$ of a common stratum are indifferent. The stratification property rules out strict comparison: were $w \succ w'$, applying Definition~\ref{def:stratification} within that stratum with the first pair $(w, w')$ and the second pair $(w', w)$ would yield $w \not\succ w'$, a contradiction, and a symmetric argument rules out $w' \succ w$. Completeness then upgrades incomparability to indifference, giving $w \sim w'$.

To show $Z$ is $\succsim$-dense, take $x, y \in X$ with $x \succ y$. By separation there are $A \in \mathcal{A}$ and $z \in A$ with $x \succsim z$, and moreover $y \not\succ z'$ for every $z' \in A$. By the observation, $z_A \sim z$, so $x \succsim z \sim z_A$ gives $x \succsim z_A$. Also $y \not\succ z_A$ (as $z_A \in A$), and completeness then forces $z_A \succsim y$: otherwise $y \succsim z_A$, and with $\neg(z_A \succsim y)$ this would give $y \succ z_A$, contradicting $y \not\succ z_A$. Therefore $x \succsim z_A \succsim y$ with $z_A \in Z$, establishing $\succsim$-density.

The cardinality statement follows: part~(i) yields a separating stratification of cardinality $|Z|$ from any $\succsim$-dense $Z$, and part~(ii) a $\succsim$-dense set of cardinality $|\mathcal{A}|$ from any separating stratification $\mathcal{A}$, so the two least cardinalities coincide.
\end{proof}

\begin{proof}[Proof of Theorem~\ref{theorem:impossibility}]
Let $\preceq^*$ denote the natural well-order on the ordinals. Let $\Omega$ denote the first uncountable ordinal and set $X := [0, \Omega)_{\preceq^*}$, the set of all countable ordinals. Since $|X| = \aleph_1 \leq |\mathbb{R}|$, there exists an injection $\varphi : X \to [0, 1]$.

Define a binary relation $\succsim$ on $X$ by
\[
x \succsim y \iff \varphi(x) \geq \varphi(y) \text{ and } x \preceq^* y,
\]
for all $x, y \in X$.

\emph{$\succsim$ is a partial order.} Reflexivity follows from $\varphi(x) \geq \varphi(x)$ and $x \preceq^* x$. For antisymmetry, $x \succsim y$ and $y \succsim x$ jointly imply $\varphi(x) = \varphi(y)$, hence $x = y$ by injectivity of $\varphi$. For transitivity, $x \succsim y$ and $y \succsim z$ imply $\varphi(x) \geq \varphi(y) \geq \varphi(z)$ and $x \preceq^* y \preceq^* z$, hence $x \succsim z$.

\emph{$\varphi$ is a Richter--Peleg representation of $\succsim$.} By antisymmetry, $x \sim y$ implies $x = y$, hence $\varphi(x) = \varphi(y)$. If $x \succ y$, then $\varphi(x) \geq \varphi(y)$ and $x \neq y$, so $\varphi(x) > \varphi(y)$ by injectivity.

\emph{Every separating stratification of $\succsim$ is uncountable.} Suppose, seeking a contradiction, that $\mathcal{A} = \{A_n\}_{n \in \mathbb{N}}$ is a countable separating stratification of $\succsim$. We derive a contradiction in three steps.

\textbf{Step 1.} \emph{Each stratum $A_n$ is countable.}

Fix $n \in \mathbb{N}$. We first observe a useful fact: for any $a, b \in A_n$ with $\varphi(a) > \varphi(b)$, we have $b \prec^* a$. This follows from the incomparability of $a$ and $b$: they are distinct, so $a \not\sim b$ by antisymmetry, and they lie in a common stratum, so neither strictly dominates the other. If we had $a \preceq^* b$ then $a \succsim b$ would follow. Thus $a \not\preceq^* b$, and so $b \prec^* a$.

Now set $r_n := \sup \varphi(A_n)$ and consider two cases.

\emph{Case 1: The supremum $r_n$ is attained.} Let $a \in A_n$ with $\varphi(a) = r_n$. For any $b \in A_n \setminus \{a\}$, injectivity of $\varphi$ yields $\varphi(b) < \varphi(a)$. By the observation above, $b \prec^* a$. Hence $A_n \subseteq [0, a]_{\preceq^*}$, which is countable since $a$ is a countable ordinal.

\emph{Case 2: The supremum $r_n$ is not attained.} There exists a sequence $(a_m)_{m \in \mathbb{N}}$ in $A_n$ with $\varphi(a_m) \to r_n$ and $\varphi(a_m) < r_n$ for all $m$. For any $b \in A_n$, $\varphi(b) < r_n$ implies there exists $m$ with $\varphi(b) < \varphi(a_m)$. By the observation above, $b \prec^* a_m$, so $b \in [0, a_m)_{\preceq^*}$. Therefore
\[
A_n \subseteq \bigcup_{m \in \mathbb{N}} [0, a_m)_{\preceq^*},
\]
which is a countable union of countable initial segments, hence countable.

\textbf{Step 2.} \emph{There exists $\bar{x} \in X$ with $z \prec^* \bar{x}$ for every $z \in \bigcup_n A_n$.}

By Step 1, $\bigcup_n A_n$ is a countable union of countable sets, hence countable. Since $\Omega$ is the first uncountable ordinal, every countable subset of $X = [0, \Omega)_{\preceq^*}$ is bounded above in $\preceq^*$ by some element of $X$. Take any such bound and let $\bar{x}$ denote its $\preceq^*$-successor (itself a countable ordinal, hence in $X$). Then $z \prec^* \bar{x}$ for every $z \in \bigcup_n A_n$.

\textbf{Step 3.} \emph{Contradiction.}

Consider the restriction of $\varphi$ to $[\bar{x}, \Omega)_{\preceq^*}$. This restriction cannot be $\preceq^*$-isotone (in the sense that $x \preceq^* y$ implies
$\varphi(x) \leq \varphi(y)$): if it were, injectivity of $\varphi$ would give
$\varphi(t) < \varphi(t+1)$ for every $t \in [\bar{x}, \Omega)_{\preceq^*}$, and since $t+1$
is the $\preceq^*$-successor of $t$, isotonicity would give $\varphi(t+1) \leq \varphi(t')$ whenever $t \prec^* t'$; the intervals $(\varphi(t), \varphi(t+1))$,
$t \in [\bar{x}, \Omega)_{\preceq^*}$, are then pairwise disjoint and nonempty. As
$[\bar{x}, \Omega)_{\preceq^*}$ is uncountable, this yields an uncountable family of disjoint open intervals in $\mathbb{R}$, which is impossible (cf.\ Example~\ref{example:lex}).

Hence there exist $x', x'' \in [\bar{x}, \Omega)_{\preceq^*}$ with $x' \prec^* x''$ and $\varphi(x'') < \varphi(x')$. We claim $x' \succ x''$. Indeed, $\varphi(x') > \varphi(x'')$ and $x' \preceq^* x''$ give $x' \succsim x''$; conversely, $x'' \succsim x'$ would require $\varphi(x'') \geq \varphi(x')$, contradicting our choice. Hence $x' \succ x''$.

Since $\mathcal{A}$ is separating, there exist $n \in \mathbb{N}$ and $z \in A_n$ such that $x' \succsim z$. From $x' \succsim z$, we have $x' \preceq^* z$. Combined with $\bar{x} \preceq^* x'$, this gives $\bar{x} \preceq^* z$. But $z \in \bigcup_n A_n$ implies $z \prec^* \bar{x}$ by Step 2, contradicting $\bar{x} \preceq^* z$.
\end{proof}

\begin{proof}[Proof of Proposition~\ref{proposition:no_embedding}] The nonempty level sets $u^{-1}(t)$, $t \in u(X)$, are pairwise disjoint. If $x \in u^{-1}(t)$ and $x' \in u^{-1}(t')$ satisfy $x \succ x'$, then $u(x) > u(x')$, that is, $t > t'$; and then no $y \in u^{-1}(t)$ and $y' \in u^{-1}(t')$ can satisfy $y' \succ y$, since that would give $t' > t$. The level sets therefore form a stratification of $\succsim$, of which those indexed by rationals are a countable subfamily.

Let $x \succ y$. Both $x$ and $y$ belong to $L(x)$, by reflexivity and by $x \succsim y$ respectively, so the set $u(L(x))$ contains $u(y)$ and $u(x)$. Since $L(x)$ is connected and $u$ is continuous, $u(L(x))$ is a connected subset of $\mathbb{R}$ and therefore contains the whole interval $[u(y), u(x)]$, which is nondegenerate because $u(x) > u(y)$ and hence contains a rational $q$. Choose $z \in L(x)$ with $u(z) = q$: then $u^{-1}(q)$ is nonempty, belongs to the family, and satisfies $x \succsim z$. Finally, no $w \in u^{-1}(q)$ satisfies $y \succ w$, since that would give $u(y) > u(w) = q \geq u(y)$. Hence $u^{-1}(q)$ separates the pair $(x,y)$ in the sense of Definition~\ref{def:separating_stratification}.
\end{proof} 

The proof of Proposition~\ref{proposition:complete_no_gap} rests on the following observation, which is where completeness enters.

\begin{lemma}
\label{lemma:chain_density}
Let $\succsim$ be a complete preorder on $X$ and let $\mathcal{C}$ be a family of increasing subsets of $X$ such that every $x \succ y$ satisfies $x \in C$ and $y \notin C$ for some $C \in \mathcal{C}$. If $\mathcal{C}$ is infinite, then $X$ has a $\succsim$-dense subset of cardinality at most $|\mathcal{C}|$. 
\end{lemma}

\begin{proof} 
We first note that $\mathcal{C}$ is a chain under set inclusion.\footnote{That is, either $C \subseteq C'$ or $C' \subseteq C$ holds for all $C, C' \in \mathcal{C}$.} Indeed, if $b \in C \setminus C'$ and $b' \in C' \setminus C$, then completeness gives $b \succsim b'$ or $b' \succsim b$; the first places $b \in C'$, since $b' \in C'$ and $C'$ is increasing, and the second places $b' \in C$.

Write $\kappa := |\mathcal{C}|$ and construct $Z \subseteq X$ as follows. For each $C \in \mathcal{C}$ possessing a $\succsim$-minimum, place one such minimum in $Z$. For each pair $D \subsetneq C$ in $\mathcal{C}$, place one element $z_{C,D} \in C \setminus D$ in $Z$. Then $|Z| \leq \kappa + \kappa^2 = \kappa$.

Let $x \succ y$ and choose $C \in \mathcal{C}$ with $x \in C$ and $y \notin C$. Suppose first that $C$ has a $\succsim$-minimum $m$. Then $x \in C$ gives $x \succsim m$, while $y \succsim m$ would place $y \in C$; completeness therefore gives $m \succ y$, and $x \succsim m \succsim y$ with $m \in Z$.

Suppose instead that $C$ has no $\succsim$-minimum. Then $x$ is not one, so some $c \in C$ satisfies $\neg(c \succsim x)$ and hence, by completeness, $x \succ c$; write $x_1 := c \in C$. Applying the separating hypothesis to the pair $(x, x_1)$ yields $D \in \mathcal{C}$ with $x \in D$ and $x_1 \notin D$. Since $x_1 \in C$ we cannot have $C \subseteq D$, so $D \subsetneq C$ by the chain property, and $y \notin D$ because $D \subseteq C$. Moreover $x_1 \in C \setminus D$, so $z := z_{C,D}$ is defined. Now $x \in D$, $z \notin D$ and $D$ increasing give $\neg(z \succsim x)$, hence $x \succ z$ by completeness; and $z \in C$, $y \notin C$ and $C$ increasing give $\neg(y \succsim z)$, hence $z \succ y$. Thus $x \succsim z \succsim y$ with $z \in Z$.
\end{proof}

\begin{proof}[Proof of Proposition~\ref{proposition:complete_no_gap}]
The inequality $\mathfrak{s}^*(\succsim) \leq \mathfrak{s}(\succsim)$ holds because $\succsim$ embeds in itself. For the converse, let $f$ be an embedding of $\succsim$ into a preorder $\succsim^*$ carrying a separating stratification $\mathcal{A}^*$ of cardinality $\kappa$. For each $A \in \mathcal{A}^*$ set \[ D(A) := f^{-1}\left(\uphull{A}\right), \qquad E(A) := f^{-1}\left(\Uphull{A}\right). \] Both are increasing subsets of $X$: if $w \in D(A)$, say $f(w) \succsim^* a$ with $a \in A$, and $w' \succsim w$, then $f(w') \succsim^* f(w) \succsim^* a$, and similarly for $E(A)$, since $\Uphull{A}$ is increasing in $\succsim^*$ by \eqref{eq:strict_composition}.

These sets separate every strict pair of $\succsim$. Let $x \succ y$. Since $f$ is an embedding, $f(x) \succ^* f(y)$, so there is $A \in \mathcal{A}^*$ with $f(x) \in \uphull{A}$ and $f(y) \notin \Uphull{A}$. If $f(y) \notin \uphull{A}$, then $D(A)$ contains $x$ and not $y$. Otherwise $f(y) \succsim^* a$ for some $a \in A$, and $f(x) \succ^* f(y) \succsim^* a$ gives $f(x) \succ^* a$ by \eqref{eq:strict_composition}, so that $E(A)$ contains $x$ and not $y$.

The family $\left\{D(A) \, \middle| \, A \in \mathcal{A}^*\right\} \cup \left\{E(A) \, \middle| \, A \in \mathcal{A}^*\right\}$ therefore separates every strict pair and has cardinality at most $2\kappa$. It must be infinite: by the chain property established at the start of the proof of Lemma~\ref{lemma:chain_density}, a finite separating family of $n$ increasing sets is a chain and therefore partitions $X$ into at most $n+1$ blocks, no two elements of a block are strictly ranked, and completeness then makes each block an indifference class, so that $X$ would have finitely many indifference classes and $\mathfrak{s}(\succsim)$ would be finite by Proposition~\ref{proposition:order_density}, contrary to hypothesis. Lemma~\ref{lemma:chain_density} therefore yields a $\succsim$-dense subset of cardinality at most $2\kappa = \kappa$, whence $\mathfrak{s}(\succsim) \leq \kappa$ by Proposition~\ref{proposition:order_density}. Minimizing over $\succsim^*$ gives $\mathfrak{s}(\succsim) \leq \mathfrak{s}^*(\succsim)$.
\end{proof}

\begin{proof}[Proof of Corollary~\ref{corollary:preorder}]
Since $\succsim$ is a preorder, it is strongly acyclic. Moreover, $\succsim^t \, = \, \succsim$. Hence, the desired equivalence follows from Theorem~\ref{theorem:main}.
\end{proof}

\begin{proof}[Proof of Corollary~\ref{corollary:countable}]
Since $\succsim^t$ is reflexive, Example~\ref{example:separating_singletons} shows that $\left\{\{x\}\mid x \in X\right\}$ is a separating stratification of $\succsim^t$; it is countable because $X$ is. Hence $\succsim^t$ is separable, and Theorem~\ref{theorem:main} provides the desired equivalence.
\end{proof}

\begin{proof}[Proof of Proposition~\ref{proposition:optimality1}]
Let $\succsim$ be a preorder on $X$ with a Richter--Peleg representation $v:X \to \mathbb{R}$. If $A=\emptyset$, the result holds using $u=v$, so suppose $A \ne \emptyset$. We can assume $v(X)\subseteq (0,1)$ without loss of generality. Define
\[
u(x) :=
\begin{cases}
1 & \exists y \in \mathcal{M}(\succsim,A): y \sim x,\\
v(x) & \exists y \in A : y \succ x,\\
1+ v(x) & \text{else,}
\end{cases}
\]
for each $x \in X$. The three cases are mutually exclusive. Cases~2 and~3 are disjoint by construction, and cases~1 and~2 cannot co-occur: if $w \in \mathcal{M}(\succsim,A)$ with $w \sim x$ and $y \succ x$ for some $y \in A$, then $y \succ x \succsim w$ gives $y \succ w$ by \eqref{eq:strict_composition}, contradicting $w \in \mathcal{M}(\succsim,A)$. So exactly one case holds and $u$ is well defined.

We now show $u$ is a Richter--Peleg representation of $\succsim$. Take $x,y \in X$. We need to show $x \sim y$ implies $u(x)=u(y)$ and $x \succ y$ implies $u(x)>u(y)$.

First assume $x \sim y$. Both $x$ and $y$ satisfy the same condition in the piecewise definition of $u$: either $u(x)=u(y)=1$, or $u(x)=v(x), u(y)=v(y)$, or $u(x)=1+v(x), u(y)=1+v(y)$. In each case $u(x)=u(y)$.

Next assume $x \succ y$. Then one of two cases holds: (I) $u(y) \in \{1, 1+v(y)\}$, or (II) $u(y)=v(y)$. In case (I), the construction of $u$ implies that $z \not\succ y$ for all $z \in A$. It follows that $z \not\succsim x$ for all $z \in A$: were $z \succsim x \succ y$ for some $z \in A$, \eqref{eq:strict_composition} would give $z \succ y$, a contradiction. Since both case~1 and case~2 in the definition of $u$ require some element of $A$ to be $\succsim$-above $x$---a maximal element indifferent to $x$ in the first, a strict dominator in the second---the alternative $x$ falls in case~3, and $u(x)=1+v(x)>1+v(y)\geq u(y)$. In case (II), $u(x) \in \{1, v(x), 1+v(x)\}$, implying $u(x)>v(y)=u(y)$. In either case, $u(x)>u(y)$.

To establish ${\mathcal{M}(\succsim, A)} = \argmax_{x \in A} u(x)$, note first that $u(x)=1+v(x)$ never holds for $x \in A$ (such an $x$ would be $\succsim$-maximal in $A$, hence fall in case~1 with $y = x$), so either $u(x) = 1$ or $u(x)=v(x)< 1$ for every $x \in A$. Then consider two cases. If $\mathcal{M}(\succsim, A) = \emptyset$, then $u(x)=v(x)$ for every $x \in A$, so $\argmax_{x \in A} u(x) = \argmax_{x \in A} v(x)$; this set is empty, since any $v$-maximizer over $A$ is $\succsim$-maximal in $A$---no $y \in A$ has $y \succ x^*$, as that would give $v(y)>v(x^*)$---and would thus belong to $\mathcal{M}(\succsim,A)$. Hence $\argmax_{x\in A}u(x)=\emptyset=\mathcal{M}(\succsim,A)$. If $\mathcal{M}(\succsim, A)\ne\emptyset$, then 
\[
\argmax_{x \in A} u(x) = \left\{x \in A \, \middle| \, u(x)=1\right\} = \left\{x \in A \, \middle| \, \exists y \in \mathcal{M}(\succsim, A): y \sim x\right\} = \mathcal{M}(\succsim, A).
\]
For the final equality, $\supseteq$ holds by reflexivity; for $\subseteq$, if $x \in A$ and $x \sim w$ with $w \in \mathcal{M}(\succsim,A)$, then any $y \in A$ with $y \succsim x$ satisfies $y \succsim w$, so $w \succsim y$ by maximality of $w$, whence $x \succsim y$; thus $x \in \mathcal{M}(\succsim,A)$.
\end{proof}

\begin{proof}[Proof of Proposition~\ref{proposition:optimality2}]
Fix $x \in X$ and let $v$ be any Richter--Peleg representation of $\succsim$ satisfying $v(X)\subseteq (0,1)$. Define
\[
u(y) :=
\begin{cases}
1 & y \sim x,\\
1+ v(y) & y \succ x,\\
v(y) & \text{else.}
\end{cases}
\]
Take $y, y' \in X$. If $y \sim y'$, then one of three cases holds: (I) $y \sim x$ and $y' \sim x$, (II) $y \succ x$ and $y' \succ x$, or (III) $y \not\succsim x$ and $y' \not\succsim x$. In case (I), $u(y)=1=u(y')$. In case (II), $u(y)=1+v(y)=1+v(y')=u(y')$. In case (III), $u(y)=v(y)=v(y')=u(y')$. In all three cases, $u(y)=u(y')$. If $y \succ y'$, then one of two cases holds: (I) $y' \succsim x$ or (II) $y' \not\succsim x$. In case (I), $y \succ y' \succsim x$ gives $y \succ x$ by \eqref{eq:strict_composition}, so $u(y)=1+v(y) > 1+v(y') \geq u(y')$. In case (II), $u(y) \geq v(y) > v(y') = u(y')$. In both cases, we have $u(y)>u(y')$. We conclude that $u$ is a Richter--Peleg representation of $\succsim$.

Now take any $A \subseteq X$ such that $x \in \mathcal{M}\left(\succsim, A\right)$. For any $y \in A$, $y \not \succ x$. Hence $u(x) = 1 \geq u(y)$.
\end{proof}

\begin{proof}[Proof of Proposition~\ref{proposition:suppes_sen}]
We use the following fact: if $\pi \in \Pi$ and $\pi y \geq y$, then $\pi y = y$. Indeed, $\pi$ permutes a finite set $F$ of coordinates and fixes the rest; the sum of $y$ over $F$ is invariant under $\pi$, and every coordinate weakly increases, so all are unchanged.

Suppose $W$ is a Richter--Peleg representation of $\succsim_{S}$. For every $\pi \in \Pi$ we have $x \succsim_{S} \pi x$ (take the permutation $\pi$) and $\pi x \succsim_{S} x$ (take $\pi^{-1}$), so $x \sim_{S} \pi x$ and $W(\pi x) = W(x)$, which is finite anonymity. If $x \geq y$ and $x \neq y$, then $x \succsim_{S} y$ via the identity permutation; were $y \succsim_{S} x$, some $\pi \in \Pi$ would satisfy $\pi y \geq x \geq y$, whence $\pi y = y$ by the fact above and therefore $x = y$, a contradiction. So $x \succ_{S} y$ and $W(x) > W(y)$, which is strong Pareto.

Conversely, suppose $W$ satisfies strong Pareto and finite anonymity. Then $\pi x \geq y$ implies $W(x) = W(\pi x) \geq W(y)$, with strict inequality unless $\pi x = y$. If $x \sim_{S} y$, applying this in both directions gives $W(x) = W(y)$. If $x \succ_{S} y$, take $\pi \in \Pi$ with $\pi x \geq y$; were $\pi x = y$, then $\pi^{-1} y = x$ would give $y \succsim_{S} x$, a contradiction, so $\pi x \neq y$ and $W(x) > W(y)$. 
\end{proof}

\end{document}